\documentclass{article}

\usepackage[normalem]{ulem} 
\usepackage{amsmath, amsfonts} 
\usepackage{amsthm} 
\usepackage{float} 
\usepackage{xcolor}
\usepackage[colorlinks = true,
            linkcolor = blue,
            urlcolor  = red,
            citecolor = red,
            anchorcolor = yellow]{hyperref}

\usepackage{booktabs}
\newcommand{\ra}[1]{\renewcommand{\arraystretch}{#1}}

\newtheorem{theorem}{Theorem}[section]

\newtheorem{remark}{Remark}[section]

 \newtheorem*{theorem*}{Theorem}
\newtheorem{example}{Example}[section]

\usepackage{graphicx}
\graphicspath{ {./Graphics/} }

\def\1{\mbox{1\hspace{-.25em}{I}}}

\def\Ex{\mathbf{E}}
\def\R{\mathbf{R}}

\def\N{\mathcal{N}}
\def\U{\mathbb{U}}
\def\V{\mathrm{v}}
\def\E{\mathbb{E}}
\def\Pb{\mathbf{P}}

\def\H{\mathcal{H}}
\def\C{\mathbb{C}}

\title{Adjusted Win Ratio with Stratification: Calculation Methods and Interpretation}
\author{Samvel B. Gasparyan, 
	Folke Folkvaljon, 
	Olof Bengtsson, 	\and 
	Joan Buenconsejo\thanks{AstraZeneca R\&D. Corresponding author, Samvel B. Gasparyan, AstraZeneca R\&D, Pepparedsleden 1, 431 53 Mölndal, Sweden. E-mail: \href{mailto:samvel.gasparyan@astrazeneca.com}{samvel.gasparyan@astrazeneca.com}},
	Gary G. Koch\thanks{Department of Biostatistics, University of North Carolina, Chapel Hill, NC 27599-7420.}\\
	}

\date{\today}
\begin{document}

\maketitle

\begin{abstract}
The win ratio is a general method of comparing locations of distributions of two independent, ordinal random variables, and it can be estimated without distributional assumptions. In this paper we provide a unified theory of win ratio estimation in the presence of stratification and adjustment by a numeric variable. Building step by step on the estimate of the crude win ratio we compare corresponding tests with well known nonparametric tests of group difference (Wilcoxon rank-sum test, Fligner-Plicello test, Cochran-Mantel-Haenszel test, test based on the regression on ranks and the rank ANCOVA test). We show that the win ratio gives an interpretable treatment effect measure with corresponding test to detect treatment effect difference under minimal assumptions. 

\end{abstract}


\tableofcontents

\listoffigures
\listoftables

\newpage
\section{Introduction}
The win ratio as a measure for analyzing clinical endpoints was suggested in \cite{Poc2012}, \cite{Wang2016} to handle composite endpoints where the components are not clinically equivalent. For example, in heart failure (HF) trials the primary endpoint of interest is usually time to the composite of cardiovascular death (CVD) or a heart failure hospitalization (HFH), whichever happens first for an individual. By combining these two into one composite endpoint, we disregard the fact that a hospitalization for heart failure is clinically different from cardiovascular death. To overcome this issue an order is introduced between the components of the composite endpoint and it is analyzed as an ordinal variable. The win ratio intends to introduce an appropriate statistical approach to analyze such endpoints. The idea is to compare the outcomes from two distributions and assign the values ``win", ``loss" or ``tie" to these comparisons based on the value of the distribution of interest being correspondingly ``better", ``worse" or ``equal" to the value from the other distribution. The advantage of such approach is that in very general situations a comparison can be defined. For example, if subjects are followed an equal period of time until an HFH or a CVD happens, an order can be introduced by treating censoring as being better (in terms of benefit to patients) than HFH which in turn is better than CVD, while subjects experiencing an event of the same type can be compared using the time of the event (later is better). Hence two groups of patients receiving different treatments can be compared using this ordering, and the benefit of one treatment against the other can be estimated using the win ratio. Recent years saw more applications of win ratio in clinical trials as a part of prespecified testing hierarchy. To give two examples, the recently announced EMPULSE trial (registration number NCT04157751 in ClinicalTrials.gov) is a multicentre, randomized, double-blind, 90-day superiority trial in patients hospitalized for acute heart failure, where the primary endpoint is defined as a hierarchical composite of time to death, number of HF hospitalizations, time to first HFH and change in a KCCQ-CSS (clinical summary score of the Kansas City cardiomyopathy questionnaire) from baseline after 90 days of treatment (see Section \ref{WRA}). On the other hand, a large-scale CV (cardiovascular) outcome trial DAPA-HF (registration number NCT03036124 in ClinicalTrials.gov) in patients with heart failure with reduced ejection fraction (HFrEF) used the win ratio for analyzing patient reported symptoms scores as the third secondary endpoint. Application of the win ratio approach in the DAPA-HF trial will be the central topic of the Section \ref{DAPA}. 

In parallel, statistical methods for analyzing the win ratio started to gain more attention. In \cite{Poc2012} a confidence interval was constructed only for the so-called {\it matched win ratio}, which uses a restrictive definition of the win ratio. For the general definition of the win ratio, \cite{Wang2016} constructed a confidence interval using the bootstrap approach. \cite{Dong2016} gave an analytical approach for construction of a confidence interval and corresponding hypothesis testing using logarithmic asymptotic distribution of the win ratio. In subsequent papers, the authors provided generalization of the win ratio for the stratified analysis in \cite{Dong2018}, interpretation of the win ratio (including the definition of estimands for the win ratio as described in ICH E9 (R1) addendum on estimands) and handling of ties in \cite{Dong2019}. The papers \cite{Luo2015}, \cite{Luo2017} derive an alternative standard error estimate for the win ratio using counting process methods. Usually in the time-to-event setting the follow-up time and the outcome at the end of the follow up are used to define the ordering, as was initially described in \cite{Poc2012}, which means that the censoring is not used in the traditional sense of having incomplete observations. \cite{Oak2016}, following \cite{Efron1967}, introduces a win ratio estimate for the censored observations. In our setting, time will be fixed and will not be used in defining the order between the outcomes.

Almost all existing analytical solutions for the standard error estimation for the win ratio use the theory of U-statistics developed in the seminal paper \cite{Hoef1948}, and the relationship of the win ratio and the Mann-Whitney test statistic is apparent (see \cite{Bebu2015}). In this article we will further explore this relationship and will reformulate the results of the generalized Mann-Whitney statistic (stratified and adjusted for a numeric baseline covariate), developed extensively in the papers \cite{Davis1968}, \cite{Puri1971}, \cite{Land1978}, \cite{Koch1982}, \cite{Koch1998}, \cite{Kaw2011}, to account for this new change in concepts. 

The win ratio is defined as an odds of the win probability. First, the win probability is introduced as the theoretical probability of one, in general, ordinal random variable being greater than a second ordinal random variable under the condition that these random variables are independent. Several examples illustrate how this theoretical probability can be calculated if the underlying distributions of these random variables are known. Then a crude estimate of the win probability, called win proportion, is introduced, and a simulation shows the convergence of the win proportion to the win probability. Building step-by-step on the crude win proportion, continuous baseline covariate adjustment and stratification is introduced. In each step, the test based on the win probability is compared with well-known tests for group difference (Wilcoxon rank-sum test, Fligner-Plicello test, Cochran-Mantel-Haenszel test, test based on the regression on ranks and the rank ANCOVA test).

\paragraph{Outline}
Section \ref{I} gives the the definition of the win probability, its interpretation, the estimation and the construction of confidence intervals. In this section we also discuss the applied problem where the win ratio approach can be used.
Section \ref{II} generalizes the win proportion for the stratified analysis and adjustment with a numeric covariate.
Section \ref{III} compares the tests based on the win probability with other non-parametric tests.
Section \ref{DAPA} applies the theory to the analysis of symptoms scores in DAPA-HF landmark trial.

\section{Win probability (WP)}\label{I}

In this section we will define and investigate the properties of the non-adjusted (crude) win probability. Non-adjusted in our setting means that we observe only the response variables without predictors. In the following sections, the analysis time is fixed.

Suppose we have two groups of subjects receiving different treatments. The first group receives placebo, the second group receives an active treatment. At some prespecified timepoint a measurement  for the primary variable of interest is done and the following values are obtained
\begin{align}\label{Eq1}
Y_1=(y_{11},\cdots,y_{1n_1}),\ \ Y_2=(y_{21},\cdots,y_{2n_2}) 
\end{align}
where $n_2$ is the number of subjects in the active treatment group and $n_1$ is the number of subjects in the placebo group. We consider only the case when only a single measurement per subjects is done and there are no missing values. The measurement values are, in general, ordinal - they have a natural ordering, that is, the values can be compared, but, unlike the numeric values, there is no distance defined between values. We take the convention that higher values correspond to better outcome. We assume that $y_{2j},\ \ j=1,\cdots,n_2$ are an i.i.d. (independent and identically distributed) sample from the distribution of the random variable $\eta$ and $y_{1i},\ \ i=1,\cdots,n_1$ are an i.i.d sample from $\xi.$ We additionally require that $\eta$  and $\xi$ be independent.

\subsection{Definition and interpretation of WP} \label{WP}

To characterize the treatment effect of the active group in comparison to the placebo group we introduce the ``win probability" of the active treatment against the placebo as
\begin{align}\label{Eq2}
\Pb(\eta>\xi).
\end{align}
In this case $P(\eta>\xi)>\frac{1}{2}$ favors the active treatment, whereas  $P(\eta>\xi)<\frac{1}{2}$ favors placebo, with no treatment difference in the case of $P(\eta>\xi)=\frac{1}{2}.$ Our goal will be to test the hypothesis of whether there is a treatment effect difference between the active group and the placebo group based on the win probability. Before proceeding with the statistical analysis, we describe an example of how the win probability can be interpreted.   

\begin{example}\label{Ex1}
Suppose that $\xi$ has a uniform distribution $\xi\sim\U[0,a],\,a>0.$ The random variable $\eta$ is independent of $\xi$ and has a uniform distribution, shifted by a non-negative number $\delta\,(0\leq\delta\leq a),$ that is $\eta\sim\U[\delta,a+\delta].$ Then, it follows that
\begin{align*}
\Pb(\eta>\xi)=\frac{1}{2}+\frac{(2a-\delta)\delta}{2a^2}.
\end{align*}
If $\delta=0,$ meaning that the random variables have the same distribution $\U[0,a],$ then $\Pb(\eta>\xi)=\frac{1}{2}.$ Figure \ref{P1} below shows the probability density function (pdf) of the random variable $\eta-\xi$ in the case when $\delta=0.$ The probability $\Pb(\eta-\xi>0)$ is the area under the curve to the right of the origin. The pdf is symmetric, therefore the probability of the difference being positive is $\frac{1}{2}$. 

\begin{figure}[H]\caption{Difference of i.i.d. uniform random variables}
\includegraphics[width=10cm]{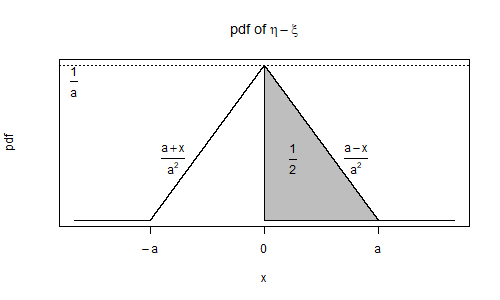}\label{P1}
\end{figure}

If $0<\delta<a,$ then the pdf of $\eta-\xi$ will be shifted to right by $\delta.$ The probability $\Pb(\eta-\xi>0)$ will be the grey area in Figure \ref{P2}, which is $\frac{1}{2}$ plus the area of the trapezoid over the interval $[0,\delta],$ calculated as

\begin{align*}
\Pb(\eta>\xi)=\frac{1}{2}+\frac{1}{2}\delta\left(\frac{a-\delta}{a^2}+\frac{1}{a}\right)=\frac{1}{2}+\frac{(2a-\delta)\delta}{2a^2}.
\end{align*}

\begin{figure}[H]\caption{Difference of uniform distributions - shifted}
\includegraphics[width=10cm]{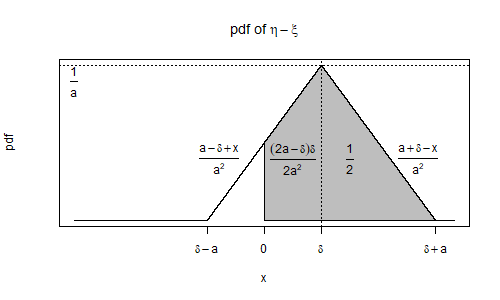}\label{P2}
\end{figure}
Remembering that $\delta=\Ex(\eta-\xi),$ we see that the difference between the win probability and $\frac{1}{2}$ is the difference in probability corresponding to the change from 0 of the mean difference of random variables.  So, the mean difference $\delta\in[0,a]$ of random variables corresponds to the following deviation of the win probability from $\frac{1}{2}$
\begin{align*}
\Pb(\eta>\xi)-\frac{1}{2}=\frac{(2a-\delta)\delta}{2a^2}.
\end{align*}
We see that there is a quadratic increase in the win probability, which will attain its maximal value for the shift $\delta=a.$ In the latter case $\Pb(\eta>\xi)=1,$ since the intervals where the densities are not 0 are completely separated and the density of $\eta$ is entirely to the right of the density of $\xi$. Hence the win probability gives a quantitative probabilistic interpretation to the mean difference. 
\end{example}

In Example \ref{Ex1} we had two identically distributed random variables, one of which had shifted mean value. The next example shows that the same interpretation is true for normally distributed random variables, even if the variance of random variables is different as well. 

\begin{example}\label{Ex2}
Suppose that $\xi\sim\N(m_1,\sigma_1^2)$ and $\eta\sim\N(m_2,\sigma_2^2)$ are independent, normally distributed random variables. Then $\eta-\xi\sim\N(m_2-m_1,\sigma^2_1+\sigma^2_2).$ The probability $\Pb(\eta>\xi)$ is the grey region in Figure \ref{P3} below (shown for the case $m_2>m_1$),

\begin{figure}[H]\caption{Difference of two normal distributions}
\includegraphics[width=10cm]{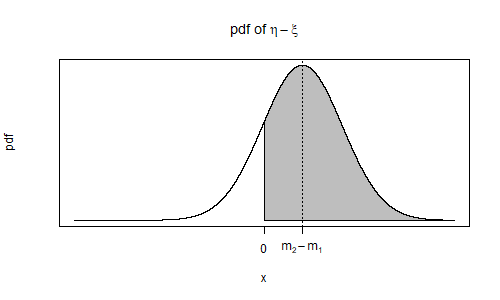}\label{P3}
\end{figure}
and it can be calculated using the formula
\begin{align}\label{Eq10}
\Pb(\eta>\xi)=\Phi\left(\frac{m_2-m_1}{\sqrt{\sigma_1^2+\sigma_2^2}}\right),
\end{align}
where $\Phi(\cdot)$ is the distribution function of the standard normal random variable. For example, in the case of  $\xi\sim\N(0,1)$ and $\eta\sim\N(1,1)$, we have $\Pb(\eta>\xi)=0.76,$ and when $\xi\sim\N(0,1)$ but $\eta\sim\N(1,4)$, then $\Pb(\eta>\xi)=0.67.$ The latter shows that the increase of variance in one of the random variables reduces the win probability.
\end{example}

\begin{remark}\label{R1} The random variable $\zeta$ is called symmetric if it has the same distribution as $-\zeta$. The median of the distribution of the random variable $\zeta$ is defined as any number $\mu$ satisfying the inequalities 
$$\Pb(\zeta\geq\mu)\geq\frac{1}{2}\text{ and }\Pb(\zeta\leq\mu)\geq\frac{1}{2}.$$
The following statements immediately follow from the definitions above.
\begin{enumerate}
\item If $\zeta$ has a continuous distribution and a unique median $\mu_0,$ then $$\Pb(\zeta>\mu_0)=\frac{1}{2}.$$
\item If $\zeta$ has a continuous distribution which is symmetric, then $$\Pb(\zeta>0)=\frac{1}{2}.$$
\item For all continuous, identically distributed random variables $\xi$ and $\eta$, the win probability is equal to $\Pb(\eta>\xi)=\frac{1}{2}.$ Indeed, since $\zeta=\eta-\xi$ is symmetric, then the win probability of both random variables will be $\frac{1}{2}.$
\end{enumerate}
\end{remark}

Using Remark \ref{R1}, Examples \ref{Ex1}, \ref{Ex2} can be generalized as follows.

\begin{example}\label{Ex7} Suppose that the independent random variables $\xi$ and $\eta$ have continuous distributions.
\begin{enumerate}
\item If there is a real number $\delta$ such that the random variable $$\zeta=\eta-\xi-\delta$$ is symmetric, then $$\Pb(\eta>\xi+\delta)=\frac{1}{2}.$$ 
Hence the equality of the win probability to $\frac{1}{2},$ that is, $\theta=\frac{1}{2}$ is the same as $\delta=0$. The inequality $\delta<0$ is equivalent to $\theta<\frac{1}{2}$ and $\delta>0$ is equivalent to $\theta>\frac{1}{2}.$ 

In particular, since $\Ex(\zeta)=0$ (because $\zeta$ is symmetric, see point 2 of Remark \ref{R1}), then we have $\delta=\Ex(\eta)-\Ex(\xi),$ hence $\theta=\frac{1}{2}$ is equivalent to equality of means of these random variables. Again, the win probability can be used also for the comparisons of means, namely if $\Ex(\eta)>\Ex(\xi)$ then $\Pb(\eta>\xi)>\frac{1}{2}$ and if $\Ex(\eta)<\Ex(\xi)$ then $\Pb(\eta>\xi)<\frac{1}{2}.$
\item Suppose that the random variable $\zeta=\eta-\xi$ has a unique median $\delta_0.$ Then, again $$\Pb(\eta>\xi+\delta_0)=\frac{1}{2}.$$ Hence, $\theta=\frac{1}{2}$ is the same as the median of the difference of these random variables being 0. A positive value of $\delta_0$ means the win probability is greater than $\frac{1}{2},$ and negative $\delta_0$ means the win probability is less than $\frac{1}{2}.$
\item Suppose that there exists a real number $\delta$ such that $\xi$ and $\eta-\delta$ have the same distribution function $F(\cdot).$ This means that the distribution function of $\eta$ differs from the distribution of $\xi$ only by a shift $\delta$ (like in Example \ref{Ex1}). From point 3 in Remark \ref{R1} we get
$$\Pb(\eta-\delta>\xi)=\frac{1}{2}.$$
Thus, as in the examples above, $\theta=\frac{1}{2}$ means there is no shift in distributions, whereas $\theta>\frac{1}{2}$ expresses a positive shift and $\theta<\frac{1}{2}$ expresses a negative shift.
\end{enumerate}
\end{example}

\begin{remark}\label{R7} Example \ref{Ex7} shows that the win probability, $\theta,$ expresses, in some sense, a comparison of locations of two distributions (like the mean difference, the median of the difference of distributions or, in the case of shifted distributions, the shift). While the mentioned location comparison statistics are relevant under some assumptions, the win probability can be defined in all cases, even when the random variables are ordinal (comparison of the random variables is defined, but the difference or sum is not). Also, from \eqref{Eq10} in Example \ref{Ex2}, we see that the value of the win probability can depend on the scale parameters of the distributions (the variances) as well, which shows that the win probability can contain more information about the comparison of two distributions than just the comparison of their locations. (In this case it contains information about the spread of the distributions as well). Thus, the win probability gives more complete information about closeness (equality) of distributions.
\end{remark}
So far we have considered only random variables with continuous distributions. By modifying the definition \eqref{Eq2} of a win probability we can have a general definition of a win probability for all random variables.
\begin{example}\label{Ex3}
Consider the case of discrete random variables. Suppose that the variable $\xi$ is constant and $\xi=1,$ whereas the random variable $\eta$ takes the values  $\eta\in\{1,2\}$ with corresponding probabilities $p_{1}=0.8,\,p_{2}=0.2.$ Then
\begin{align*}
\Pb(\eta>\xi)=\Pb(\eta=2)=0.2.
\end{align*}
So the win probability of the random variable $\eta$ is less than $\frac{1}{2}.$ But if we consider the mean values of these variables we see that $\Ex(\xi)=1$ and $\Ex(\eta)=1.2.$ Clearly, the comparison of the win probability with $\frac{1}{2}$ does not reflect the difference in the mean values. The reason is the presence of ties, in other words the probability of the event $\{\xi=\eta\}$ is positive. 
\end{example}
To have consistency with the case of continuous random variables, where $\Pb(\eta>\xi)>\frac{1}{2}$ reflected the fact of $\eta$ being better than $\xi$, we will redefine the win probability as 
\begin{align}\label{Eq4}
\theta=\Pb(\eta>\xi)+0.5\Pb(\eta=\xi).
\end{align}
In Example \ref{Ex3} this redefined win probability would be
\begin{align*}
\Pb(\eta>\xi)+0.5\Pb(\eta=\xi)=\Pb(\eta=2)+0.5\Pb(\eta=1)=0.6.
\end{align*}
The important property in Remark \ref{R1} of identically distributed random variables having the win probability equal to $\frac{1}{2}$ can be extended to the case of non-continuous random variables as well. Denote $\zeta=\eta-\xi,$ which again will be symmetric, that is, $\Pb(\zeta<x)=\Pb(-\zeta<x)$ for all real values $x$. Hence from the equality $\Pb(\zeta>0)+\Pb(\zeta<0)+\Pb(\zeta=0)=1$ we get $2\Pb(\zeta>0)+\Pb(\zeta=0)=1$, and so $\Pb(\eta>\xi)+0.5\Pb(\eta=\xi)=\frac{1}{2}$. Therefore $\theta=\frac{1}{2}$ and henceforth we will use the general definition \eqref{Eq4} of the win probability.

\begin{remark}(Number Needed to Treat)\label{R16} Consider the case when the independent random variables $\eta$ and $\xi$ are Bernoulli random variables with the probability of success being, correspondingly, $p$ and $q$. Then, the win probability of $\eta$ against $\xi$ would be
\begin{align*}
\theta&=\Pb(\eta=1,\xi=0)+\frac{1}{2}(\Pb(\eta=0,\xi=0)+\Pb(\eta=1,\xi=1))=\\
&=p(1-q)+\frac{1}{2}(pq+(1-p)(1-q))=\frac{p-q}{2}+\frac{1}{2}.
\end{align*}
Sometimes to characterize the benefit of an active treatment over a control an NNT (number needed to treat) is calculated as the inverse of the absolute benefit of intervention (see, for example, \cite{Chat1996})
\begin{align*}
NNT = \frac{1}{p-q}.
\end{align*}
Therefore, the NNT can be calculated using the win probability as follows
\begin{align*}
NNT = \frac{1}{2\theta-1}.
\end{align*}
This formula can be used in a more general setting as well when $\eta$ and $\xi$ are any two ordinal random variables. To calculate the estimated NNT we need to replace the win probability with its estimate.
\end{remark}

\subsection{WP estimation} \label{WPE}
Consider the estimation problem of the win probability \eqref{Eq4} for independent, in general ordinal, random variables $\eta$ and $\xi$ using the samples \eqref{Eq1}. The events $\{ \xi<\eta\}$ and $\{ \xi=\eta\}$ are defined on the set of all possible values $(\xi,\eta).$ Hence, to estimate the probability \eqref{Eq4} we can use the estimator
\begin{align}\label{Eq5}
\hat\theta_N=\frac{1}{n_1n_2}\sum_{i=1}^{n_1}\sum_{j=1}^{n_2}(\1\{y_{2j}>y_{1i}\}+0.5\1\{y_{1i}=y_{2j}\}).
\end{align}
Here $\1$ is an indicator taking the value 1 if the corresponding specification is satisfied, or 0 otherwise, and $N=(n_1+n_2)$. There are $n_1n_2$ possibilities of comparing a component of $Y_2$ to a component of $Y_1.$ For each comparison we can have three results - a ``win" for the active group if $y_{2j}>y_{1i}$, a ``loss" if $y_{2j}<y_{1i}$ or a ``tie" if $y_{1i}=y_{2j}$. The estimator \eqref{Eq5} counts the number of wins and one half of the number of ties over all possible combinations. The statistic $n_1n_2\hat\theta_N$ is known as the {\it Mann-Whitney statistic}. The estimator $\hat\theta_N$ is a simple frequency estimator to estimate the probability of success in a trinomial trial. By {\it the law of large numbers} this estimator tends to the win probability $\theta,$ when $N\rightarrow+\infty.$  We call the estimator \eqref{Eq5} {\it the win proportion} of the active group against the placebo group. Modifying the win proportion we can write
\begin{align}\label{Eq8}
\hat\theta_N=\frac{1}{n_2}\sum_{j=1}^{n_2}\frac{1}{n_1}\sum_{i=1}^{n_1}(\1\{y_{2j}>y_{1i}\}+0.5\1\{y_{1i}=y_{2j}\})=\frac{1}{n_2}\sum_{j=1}^{n_2}p_{j},
\end{align}
where
\begin{align}\label{Eq6}
p_j=\frac{1}{n_1}\sum_{i=1}^{n_1}(\1\{y_{2j}>y_{1i}\}+0.5\1\{y_{1i}=y_{2j}\}),\ \ j=1,\cdots,n_2.
\end{align}
We call these quantities {\it the individual win proportions} of subject $j$ in the active group against the placebo group. In the same way, we can define the win proportion of an individual $i$ in the placebo group against the active group as
\begin{align}\label{Eq15}
q_i=\frac{1}{n_2}\sum_{j=1}^{n_2}(\1\{y_{2j}<y_{1i}\}+0.5\1\{y_{1i}=y_{2j}\}),\ \ i=1,\cdots,n_1.
\end{align}
It is easy to see that
\begin{align}\label{Eq9}
\hat\theta_N=\frac{1}{n_2}\sum_{j=1}^{n_2}p_{j}=1-\frac{1}{n_1}\sum_{i=1}^{n_1}q_{i}.
\end{align}
Thus, the mean of the $p_j$ in \eqref{Eq6} and the mean of the $q_i$ in \eqref{Eq15}, respectively estimate the probabilities $\theta$ and $(1-\theta)$ (see \eqref{Eq4}). Therefore the samples \eqref{Eq1} of independent, in general ordinal random variables, can be replaced by numeric samples  (see Appendix V in \cite{Koch1998})
\begin{align}\label{Eq7}
Y_2^0=(p_1,p_2,\cdots,p_{n_2}), \ \ Y_1^0=(q_1,q_2,\cdots,q_{n_1}).
\end{align}
The following theorem provides the asymptotic normality of the win proportion as an estimator for the win probability.
\begin{theorem}\label{T8}
If $\frac{n_1}{N}\rightarrow\lambda\in(0,1),$ as $N\rightarrow+\infty,$ and independent random variables $\xi$ and $\eta$ have continuous distributions, then
\begin{align*}
\sqrt{N}(\hat\theta_N-\theta)\Longrightarrow\N\left(0,\frac{\sigma_{10}^2}{\lambda}+\frac{\sigma_{01}^2}{1-\lambda}\right),\ \ \text{as }N\rightarrow+\infty,
\end{align*}
where 
\begin{align*}
&\sigma_{01}^2=\C ov(\1(\xi<\eta),\1(\xi'<\eta))=\Pb(\xi<\eta,\xi'<\eta)-\Pb(\xi<\eta)^2,\\
&\sigma_{10}^2=\C ov(\1(\xi<\eta),\1(\xi<\eta'))=\Pb(\xi<\eta,\xi<\eta')-\Pb(\xi<\eta)^2.\\
\end{align*}
Here $\xi,\xi',\eta,\eta'$ are independent. $\xi,\xi'$ have the same distribution and $\eta,\eta'$ have the same distribution. 
\end{theorem}
The proof of Theorem \ref{T8} is based on the theory of U-statistics developed in \cite{Hoef1948}. See also the Theorem 12.6 in \cite{VdV2000}. For generalizations see \cite{Puri1971}.

The following theorem from \cite{Koch1998} gives estimates for the variances $\sigma_{10}^2,\sigma_{01}^2$. 
\begin{theorem}\label{T1}
Suppose that the random variables $\xi$ and $\eta$ are independent and not constant. The estimator $\hat\theta_N$ (see \eqref{Eq5}) of the win probability \eqref{Eq4} is asymptotically normal with
\begin{align*}
\frac{\hat\theta_N-\theta}{\sqrt{\frac{\V ar(Y_2^0)}{n_2}+\frac{\V ar(Y_1^0)}{n_1}}}\Longrightarrow\N(0,1),\ \ \text{as }n_1\rightarrow+\infty,\,n_2\rightarrow+\infty,
\end{align*}
where $\V ar(Y^2_0)$ and $\V ar(Y^1_0)$ denote the estimates of the variances
\begin{align*}
\V ar(Y_2^0)=\frac{1}{n_2-1}\sum_{j=1}^{n_2}(p_j-\hat\theta_N)^2,\ \ \V ar(Y_1^0)=\frac{1}{n_1-1}\sum_{i=1}^{n_1}(q_i-(1-\hat\theta_N))^2.
\end{align*}
\end{theorem}
The proof of Theorem \ref{T1} can be found in \cite{Brun2000}, where the theorem is formulated slightly differently (see Theorem \ref{T3}). In Section \ref{CWS}, we will show the equivalence of the formulations of these theorems.

Since from \eqref{Eq9} we see that $\hat\theta_N=\frac{1}{n_2}\sum_{j=1}^{n_2}p_{j}$ and $1-\hat\theta_N=\frac{1}{n_1}\sum_{i=1}^{n_1}q_{i},$ Theorem \ref{T1} allows construction of a confidence interval for the win probability \eqref{Eq4} using only the mean values and variances of the samples \eqref{Eq7}.

\begin{example}\label{Ex4}
Suppose that $\xi\sim\N(2,4^2)$ and $\eta\sim\N(4,2^2)$ are independent, normally distributed random variables. The win probability is (see formula \eqref{Eq10}) $\theta=0.673.$ If we sample $n_1=100$ numbers from $\xi$ and $n_2=500$ numbers from $\eta$, then using \eqref{Eq5} and Theorem \ref{T1} we can calculate the win proportion $\hat\theta_N=0.64332$ and its standard error $se=0.0357$. The plot below shows the convergence of the win proportion to the win probability when $n_1=100$ is fixed and $n_2$ is changing from 1 to 500 in a single random sample.

\begin{figure}[H]\label{P4}\caption{Convergence to win probability illustrated by a random sample}
\includegraphics[width=10cm]{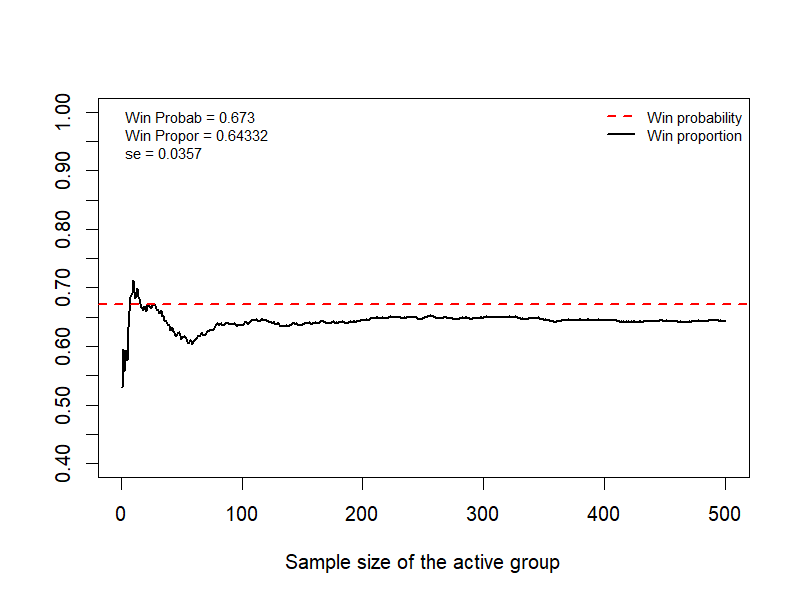}
\end{figure}
\end{example}
The left side of this figure corresponds to the case when $n_1=100$ and $n_2=1,$ hence the estimator (the win proportion) is far from the ``unknown" true parameter $\theta$ (the win probability). As we move to the right along the horizontal axis, the sample size of the active group increases to 500, and the rightmost part of the figure corresponds to the case of $n_1=100$ and $n_2=500$. Therefore we see, that by keeping the sample size of the placebo group constant and gradually increasing the sample size of the active group, the performance of the win proportion improves, and it becomes closer to the win probability. The same would be true if we were to fix the sample size of the active group and gradually increase the sample size of the placebo group. This figure illustrates the importance of the condition in Theorem \ref{T1} that the sample sizes of both treatment groups need to tend to infinity.

\subsection{Comparison to Wilcoxon rank-sum statistic}\label{CWS}
In this section we give details about the well known relationship between the Mann-Whitney statistic described in Section \ref{WPE} and the {\it Wilcoxon two sample rank-sum test}. Consider the combined sample of length $N=n_1+n_2$ of response values across the two treatment groups
\begin{align}\label{Eq18}
Y=(y_{11},\cdots,y_{1n_1}, y_{21},\cdots,y_{2n_2}). 
\end{align}
Denote by $(R_{11},\cdots,R_{1n_1},R_{21},\cdots,R_{2n_2})$ the ranks of the sample $Y.$ To get the ranks we need to order the values $y_{ij}$ increasingly. If all the values are different, then the smallest value will have the rank 1 and the biggest value will have the highest rank equal to $N$. If several values are equal, then each one of them will get the mean value of their ranks. For example, if we have the numbers $(3,3,2,1,4,4,4,4,4)$ then their ranks will be $(3.5,3.5,2,1,7,7,7,7,7)$. Consider the sum of the ranks of the second (active treatment) group
\begin{align*}
W= \sum_{j=1}^{n_2}R_{2j}.
\end{align*}
This rank-based statistic is called the {\it Wilcoxon rank-sum statistic}. If we return to our concepts of ``winning" and ``losing", then, for the sample $Y$ in \eqref{Eq18} with no ties, by subtracting 1 from the rank of the value, we will get the number of wins of that value in the entire sample. That is, $y_{2j}$ wins against $R_{2j}-1$ numbers in the combined sample $Y.$ For example, if the rank of a value is 5, then this value wins against 4 other values. Now, if we introduce ranks in the second group separately, $(\tilde R_{21},\cdots, \tilde R_{2n_2}),$ in the same way, $\tilde R_{2j}-1$ will show the number of wins of the value $y_{2j}$ against its own group. Hence, if we take the rank of the value in the entire sample and subtract the rank of the value in its own sample , $R_{2j}-\tilde R_{2j},$ we will get the number of wins of that value against the other group. Hence, if we take sum of the ranks of the second group in the entire sample and subtract the sum of the ranks of the second group ranked separately, then we will get the total number of wins of the second group against the first group. Hence,
\begin{align}\label{Eq21}
n_1n_2 \hat \theta_N=\sum_{j=1}^{n_2}(R_{2j}-\tilde R_{2j}).
\end{align}
The sum of all ranks in the second group is always $\frac{n_2(n_2+1)}{2}$ as it is the sum of numbers $1,\cdots,n_2$. Thus,
\begin{align}\label{Eq20}
n_1n_2 \hat \theta_N=W-\frac{n_2(n_2+1)}{2},
\end{align}
which is the well known relationship between the {\it Wilcoxon rank-sum statistic} $W$ and the {\it Mann-Whitney} statistic.
\begin{remark}\label{R3} 
Moreover, the relationship \eqref{Eq20} (therefore also \eqref{Eq21}) is true even in the situations where there are ties. Indeed, a rank corresponds to the number of wins + 1/2 the number of ties and added 1 in the corresponding set being ranked. Therefore, assigning the mean value of the ranks to the equal values is equivalent to adding the half of all ties to the sum of wins (see the definition in \eqref{Eq5}).
\end{remark}
The following theorem is due to \cite{Brun2000}.
\begin{theorem}\label{T3}
Suppose that the random variables $\xi$ and $\eta$ are independent and not constant. The estimator \eqref{Eq5} of the win probability \eqref{Eq4} is asymptotically normal and
\begin{align*}
\frac{\hat\theta_N-\theta}{\hat\sigma_N}\Longrightarrow\N(0,1),\ \ \text{as }n_1\rightarrow+\infty,n_2\rightarrow+\infty,
\end{align*}
where the variance (or squared standard error) is equal to 
\begin{align*}
\hat\sigma_N^2&=\frac{1}{n_1(n_1-1)n_2^2}\sum_{i=1}^{n_1}\left(R_{1i}-\tilde R_{1i}-\bar R_1+\frac{n_1+1}{2}\right)^2+\\
&+\frac{1}{n_2(n_2-1)n_1^2}\sum_{j=1}^{n_2}\left(R_{2j}-\tilde R_{2j}-\bar R_2+\frac{n_2+1}{2}\right)^2,
\end{align*}
and $\bar R_1=\frac{1}{n_1}\sum_{i=1}^{n_1}R_{1i}$ and $\bar R_2=\frac{1}{n_2}\sum_{j=1}^{n_2}R_{2j}$.
\end{theorem}
Based on previous considerations, we can write the win proportion (see \eqref{Eq6}) of each subject in the active group as $p_j=\frac{R_{2j}-\tilde R_{2j}}{n_1}.$ On the other hand (see \eqref{Eq20}),
\begin{align*}
\frac{1}{n_1}\left(\bar R_2-\frac{n_2+1}{2}\right)=\frac{1}{n_1n_2}\sum_{j=1}^{n_2}R_{2j}-\frac{(n_2+1)n_2}{2n_1n_2}=\frac{1}{n_1n_2}\sum_{j=1}^{n_2}\left(R_{2j}-\tilde R_{2j}\right),
\end{align*}
and finally, using \eqref{Eq21} and \eqref{Eq8}, we get
\begin{align*}
\frac{1}{n_1n_2}\sum_{j=1}^{n_2}\left(R_{2j}-\tilde R_{2j}\right)=\frac{1}{n_2}\sum_{j=1}^{n_2}p_j=\hat\theta_N,
\end{align*}
which means that
\begin{align*}
\frac{1}{n_2(n_2-1)n_1^2}\sum_{j=1}^{n_2}\left(R_{2j}-\tilde R_{2j}-\bar R_2+\frac{n_2+1}{2}\right)^2=\frac{1}{n_2(n_2-1)}\sum_{j=1}^{n_2}\left(p_j-\hat\theta_N\right)^2.
\end{align*}
Using the same arguments we get
\begin{align*}
\hat\sigma_N^2=\frac{1}{n_1(n_1-1)}\sum_{i=1}^{n_1}\left(q_i-(1-\hat\theta_N)\right)^2+\frac{1}{n_2(n_2-1)}\sum_{j=1}^{n_2}\left(p_j-\hat\theta_N\right)^2,
\end{align*}
which is the same variance as in Theorem \ref{T1}. Thus, the classical relationship between the Mann-Whitney statistic and Wilcoxon rank-sum statistic holds for the tied values as well. The estimator of the win probability is
\begin{align}\label{Eq23}
\hat\theta_N=\frac{1}{n_2}\sum_{j=1}^{n_2}p_j=\frac{1}{n_1n_2}\sum_{j=1}^{n_2}\left(R_{2j}-\tilde R_{2j}\right)
\end{align}
and the standard error of this estimate can be easily calculated using the means and the variances of the sample \eqref{Eq7}.

\begin{remark}\label{R8} 
Formula \eqref{Eq23} and that for $\hat\sigma_N^2$ in Theorem \ref{T3} are computationally less intense  to implement than formula \eqref{Eq8}. Therefore, formula \eqref{Eq23} is more preferable. For formulas \eqref{Eq21}, \eqref{Eq20} and \eqref{Eq23} see the last paragraph of {\it Appendix VII: Computations} in \cite{Koch1998}.
\end{remark}

\subsection{Hypothesis testing for win ratio} \label{WR}

Alongside the win probability, which was introduced to characterize the difference in treatment effects in two treatment groups, we will consider also the {\it the win ratio} which is the odds of winning of the active group against the placebo group. The win ratio, $\kappa,$ is defined as the win probability divided by the probability of loss (having ties equally divided between wins and losses), that is
\begin{align}\label{Eq12}
\kappa=\frac{\theta}{1-\theta}=\frac{\Pb(\eta>\xi)+0.5\Pb(\eta=\xi)}{\Pb(\eta<\xi)+0.5\Pb(\eta=\xi)}.
\end{align}
A win ratio $\kappa=1$ corresponds to the win probability being equal to $\theta=\frac{1}{2}.$ Hence, $\kappa>1$ corresponds to having a positive treatment effect. As we saw in the previous examples, a positive mean difference corresponds to $\kappa>1.$ Now consider the setting of time-to-event analysis.

\begin{example}\label{Ex5}
Consider survival times in two treatment groups when there is no censoring. $\eta$ is the survival time in the active group, whereas $\xi$ is the survival time in the placebo group. Suppose that these two independent random variables follow exponential distributions with parameters $\varphi$ and $\lambda$ correspondingly,
\begin{align*}
\eta\sim\E(\varphi),\ \ \xi\sim\E(\lambda).
\end{align*}
In other words, the hazard functions in the two groups are constant. In this case the win probability of the active group is the probability of having longer survival time in the active group. It can be seen that the win probability is
\begin{align*}
\theta=\Pb(\eta>\xi)=\frac{\lambda}{\lambda+\varphi}=\frac{\frac{1}{\varphi}}{\frac{1}{\lambda}+\frac{1}{\varphi}}=\frac{\Ex(\eta)}{\Ex(\xi)+\Ex(\eta)}.
\end{align*}
Hence, in this example the win ratio is the inverse of the hazard ratio,
\begin{align*}
\kappa=\frac{1}{(\varphi/\lambda)}=\frac{\Ex(\eta)}{\Ex(\xi)}=\frac{\lambda}{\varphi}.
\end{align*}
This means that if the hazard ratio of active treatment versus placebo is smaller than 1, then the win ratio is greater than 1. In other words, to have better survival time the treatment group needs to have smaller hazard.
\end{example}
Example \ref{Ex5} shows that in the time-to-event setting without censoring and with constant hazards, the win ratio is the same as the inverse of the hazard ratio. This relationship between the win ratio and the hazard ratio remains true under the more general proportional hazards assumption.

\begin{example}\label{Ex6} 
Suppose again that $\eta$ is the survival time in the active group, whereas $\xi$ is the survival time in the placebo group. Consider the survival functions of these independent random variables,
$$S_\xi(t)=\Pb(\xi\geq t),\ \ S_\eta(t)=\Pb(\eta\geq t).$$
We assume that the hazard functions of these random variables are proportional (the hazard ratio is constant), which means that $S_\eta(t) = S_\xi(t)^{HR},$ where $HR$ is the hazard ratio of random variables $\eta$ and $\xi$.  Also in this case, $\theta= \frac{1}{1+HR}$ and $\kappa=\frac{\theta}{1-\theta}=\frac{1}{HR}.$
\end{example}

The treatment effect comparison can be tested by the following hypothesis for the win ratio
\begin{align}\label{Eq13}
\H_0: \ \ \kappa=1\,  \text{ against } \H_1: \ \ \kappa\neq1.
\end{align}

Theorems \ref{T1}, \ref{T3} allow construction of an asymptotic test of level $\alpha\in(0,1).$ Indeed, for example from Theorem \ref{T3}, we have
\begin{align*}
\Pb(\theta\in[\hat\theta_N-C_\alpha\hat \sigma_N,\hat\theta_N+C_\alpha\hat \sigma_N])\rightarrow1-\alpha, 
\end{align*}
where $C_\alpha$ is the $1-\frac{\alpha}{2}$ quantile of the standard normal distribution. Hence the null hypothesis $\theta=\frac{1}{2}$ will be rejected if 
$$\frac{1}{2}\notin[\hat\theta_N-C_\alpha\hat \sigma_N,\hat\theta_N+C_\alpha\hat \sigma_N].$$ 
The asymptotic type I error of this test will be $\alpha.$ Since $\kappa=\frac{\theta}{1-\theta},$ an estimator for the win ratio is 
\begin{align*}
\hat\kappa_N=\frac{\hat\theta_N}{1-\hat\theta_N}.
\end{align*}
Since the function $f(x)=\frac{x}{1-x}$ is an increasing function for $x\in(0,1),$ its application to the asymptotic confidence interval of the win probability will produce an asymptotic confidence interval for the win ratio as
\begin{align*}
[f(\hat\theta_N-C_\alpha\hat \sigma_N),f(\hat\theta_N+C_\alpha\hat \sigma_N)].
\end{align*}

\subsection{Application of the win ratio} \label{WRA}
Consider a clinical trial where an objective is to compare the change from baseline of the primary variable of interest in two treatment groups at a specific time point, say, at the end of the trial. The primary variable can be a measurement of a biomarker or a score of a questionnaire. If the subject dies during the trial, then the change from baseline of the primary variable will be missing. The assumption of missingness at random may be violated if there is a treatment effect on mortality. A more specific example is given in Section \ref{DAPA}, where the score from a symptom questionnaire is described. There is an apparent correlation between deterioration in symptoms and increased risk of death. Hence, if the subject died before the end of the trial, then we need to incorporate this information into the analysis of the primary variable, if we are interested to measure the treatment effect as it is and not in a conditional setting of subjects being alive. This can be done by considering the composite of the death and the change from baseline in the primary variable of interest, assigning the ``worst" change value to the subjects who die during the trial. In this case, combining numerical values with death will convert the primary variable of interest into an ordinal variable. Then the win ratio can be used to compare the treatment effect in two groups. 

To incorporate death into the analysis of the primary variable, we can choose several strategies. Death is always considered worse than any measured change from baseline. There are the following strategies to manage death:
\begin{enumerate}
\item Treat all deaths as equal, in other words, assign all deaths the same ordinal value.
\item Ordering among deaths is defined based on a characteristic not directly related to death, observed at or after baseline. For example, if there are measurements of the same primary variable taken prior to death, then ordering among deaths can be done based on each individual's last observed value of change from baseline while alive, meaning that if a subject has a higher change before dying than another subject who died, then this subject will have a higher ordinal value than the other subject. In a setting where there are no intermediate measurements taken, such a strategy would correspond to a baseline carried forward approach. More generally, any other measurement of a subject made during the trial, even if it is done not on the primary variable of interest, can be used to define ordering among deaths. Any clinically justified combination of characteristics measured at or after baseline could be used, as long as there is a sound rationale for the importance of such factors.
\item Ordering among deaths is defined by characteristics that are directly related to the event of death. For example, ordering among deaths can be done based on each individual's observed survival time, meaning that if a subject died later than another subject, then this subject will have a higher ordinal value than the other subject. Another example is the definition of the order based on the cause of death.
\end{enumerate} 

In Section \ref{DAPA} we will consider only the first two cases. In the first analysis all deaths will have the same ordinal value. As an alternative approach we will introduce ordering among deaths based on the last observed value of the same primary variable while alive. All other {\it intercurrent} events, for example hospitalizations, happening between the baseline measurement and the measurement done at the prespecified time point will not be included in the analysis and subsequent values of the primary variable of interest will be used. This corresponds to treatment policy strategy of handling intercurrent events, as described in \cite{ICH2019}. The treatment policy strategy is based on the Intent to Treat (ITT) principle. The purpose of the ITT principle is to measure the treatment effect on the variable of interest directly, without using the information of the intercurrent events. The treatment policy strategy cannot be applied to missingness due to death, hence the composite strategy (combining the death with observed values of the variable) to handle deaths is applied.

We note that we have flexibility in combining the numeric value with the so called ``hard" clinical outcomes. Furthermore, this approach allows the specific handling of all intercurrent events (not only death), as defined in \cite{ICH2019}. It is possible to combine intercurrent events by introducing prioritization. For example, if a subject has an event after the baseline measurement and before the end of the trial (the prespecified time point of the measurement), then the occurrence of this event can be considered clinically more important than the measurement of the primary variable at the end of the trial (which is done after the mentioned event), even if observed. Therefore instead of the measurement at the end of the trial, the occurrence of the intercurrent event can be used to inform the estimation of the treatment effect. Returning to the example of the EMPULSE trial described in the Introduction, if the subject experienced a heart failure hospitalization before day 90, then the measurement of the symptoms score at day 90 is not used in the analysis. Instead the timing of HFH (if only one HFH happened) or the total number of hospitalizations (if several HFH happened) are considered clinically more relevant, and the composite strategy is used to handle these intercurrent events. The prioritization of events is defined as follows: deaths are assigned the worst category and the order among deaths is introduced using the time to death (later is better). The next category of ordinal values is introduced using the total number of hospitalizations before day 90 (less is better) and subjects with one HFH during 90 days are compared using the time to HFH (latter is better). Finally, subjects who are alive at day 90 and did not have intercurrent HFH before that time point are compared based on their observed change from baseline of KCCQ-CSS score. 
This approach does not follow the treatment ploicy strategy, instead it uses the composite strategy to handle the intercurrent events. The choice of strategy will shape the definition of the estimand under trial, as defined in \cite{ICH2019}. The composite strategy of handling intercurrent events have more impact on the estimand of the treatment effect on the KCCQ score (unlike the treatment policy strategy), hence it does not measure the treatment effect ``purely" on the KCCQ score, but a type of ``net clinical benefit", which accounts for the most clinically relevant outcomes that occur during the 90 days follow-up time.

\section{Adjustment and stratification}\label{II}
\subsection{Adjusted win probability}\label{AWP}
The theory of analysis of covariance of ordered categorical data is extensively formulated in \cite{Koch1982} and \cite{Koch1998}. Here we will reiterate the main results of these articles.

In this section we will consider the scenario when the response variables are observed with predictor values. For example, we could be interested in the effect of the treatment having observed also some numeric baseline measurement of individuals, for example, the age. Hence the observed samples are
\begin{align}\label{Eq11}
(Y_1,X_1),\ \ (Y_2,X_2),
\end{align}
where $Y_1,Y_2$ are defined in \eqref{Eq1} and  
\begin{align*}
X_1=(x_{11},\cdots,x_{1n_1}),\ \ X_2=(x_{21},\cdots,x_{2n_2}).
\end{align*}
The response variables $Y_1,Y_2$ are in general ordinal, whereas the predictor variables $X_1,X_2$ are often numeric. Here again using the individual win proportions (as in \eqref{Eq7}), we can replace the samples \eqref{Eq11} with the samples
\begin{align*}
(Y_1^0,X_1),\ \ (Y_2^0,X_2),
\end{align*}
where the individual proportions $p_j$ and $q_i$ are calculated using the formulas \eqref{Eq6},\eqref{Eq15}, without taking into consideration the values of the covariates. 
Consider the mean values, the variances of the response variables and covariates, as well as the covariances between the response variables and covariates
\begin{align}\label{Eq16}
&\bar x_{1}=\frac{1}{n_1}\sum_{i=1}^{n_1}x_{1i},\ \ \bar x_{2}=\frac{1}{n_2}\sum_{j=1}^{n_2}x_{2j},\\
& \V ar(x_{1})=\frac{1}{n_1-1}\sum_{i=1}^{n_1}(x_{1i}-\bar x_{1})^2,\ \ \V ar(x_{2})=\frac{1}{n_2-1}\sum_{j=1}^{n_2}(x_{2j}-\bar x_{2})^2,\nonumber\\
& \C ov(x_1,y_1^0)=\frac{1}{n_1-1}\sum_{i=1}^{n_1}(q_i-(1-\hat\theta_N))(x_{1i}-\bar x_{1}),\nonumber\\
&\C ov(x_2,y_2^0)=\frac{1}{n_2-1}\sum_{j=1}^{n_2}(p_j-\hat\theta_N)(x_{2j}-\bar x_{2}).\nonumber
\end{align}
Then, the {\it adjusted win proportion} can be defined as
\begin{align}\label{Eq17}
\hat\beta_N=\hat\theta_N- \frac{\bar x_1- \bar x_2}{\frac{\V ar(x_{1})}{n_1}+\frac{\V ar(x_{2})}{n_2}}\left[\frac{\C ov(x_1,y_1^0)}{n_1} + \frac{\C ov(x_2,y_2^0)}{n_2}\right].
\end{align}

\begin{remark}\label{R5} 
For the randomized trial, $\hat\beta_N$ is an estimator for $\theta,$ since the expectation of the mean difference in covariates is zero, $\Ex(\bar x_1-\bar x_2)=0.$ 
\end{remark}
This estimator provides the win proportion of the active group against the placebo group adjusted for the mean difference in covariates. As it is well known (see, for example, \cite{Koch1998}), the adjustment for covariates provides more powerful tests for the treatment comparison through the variance reduction and accounts for possible random group differences in covariate values, so that the observed treatment effect is not driven by the random difference in covariates. In clinical trials the covariate often represents some numeric measure on patients done at baseline. Denote the covariate dependent win probability by $\beta.$ The following theorem holds
\begin{theorem}\label{T7}
The adjusted win proportion $\hat\beta_N$ is an asymptotically normal estimator for the win probability $\beta$ such that
\begin{align*}
\frac{\hat\beta_N-\beta}{\hat\sigma_N^\beta}\Longrightarrow\N(0,1),\ \ \text{as }N\rightarrow+\infty,
\end{align*}
and the applicable squared standard error is 
\begin{align*}
(\hat\sigma_N^\beta)^2=\frac{\V ar(Y_2^0)}{n_2}+\frac{\V ar(Y_1^0)}{n_1}- \frac{\left[\frac{\C ov(x_1,y_1^0)}{n_1} + \frac{\C ov(x_2,y_2^0)}{n_2}\right]^2}{\frac{\V ar(x_{1})}{n_1}+\frac{\V ar(x_{2})}{n_2}}.
\end{align*}
\end{theorem}

\begin{remark}\label{R14} Theorem \ref{T7} is applicable to any numeric covariate including dummy covariates with the values $0$ and $1.$ The method can be extended to the case of ordinal covariates by replacing the mean difference $(\bar x_1-\bar x_2)$ by a win proportion $\hat\theta_{N,x}$ defined by pairwise comparisons of the values of covariates.
\end{remark}

\subsection{Stratified win probability}\label{SWP}

In the stratified analysis, we suppose that each treatment group is divided into two separate subgroups, called strata (we are considering only the case of two strata). The measurements of subjects in different strata have different distributions, even inside the same treatment group. Therefore the model assumes that the measurements of each treatment group are characterized by two random variables each (here as before, $\eta$ denotes the active group, whereas $\xi$ denotes the placebo group)
\begin{align*}
(\xi,\xi'),\ \ (\eta,\eta').
\end{align*}
The samples from these random variables are denoted correspondingly (2 is for the active group, 1 for the placebo group) 
\begin{align}\label{Eq30}
&Y_1=(y_{11},\cdots,y_{1n_{11}}),\ \ Y_1'=(y_{11}',\cdots,y_{2n_{12}}'),\ \ n_1=n_{11}+n_{12},\\ 
&Y_2=(y_{21},\cdots,y_{2n_{21}}),\ \ Y_2'=(y_{21}',\cdots,y_{2n_{22}}'),\ \ n_2=n_{21}+n_{22},\,N=n_1+n_2.\nonumber
\end{align}
For each stratum, the win probability is defined as
\begin{align}\label{Eq31}
\theta=\Pb(\eta>\xi)+0.5\Pb(\eta=\xi)\text{ and } \theta'=\Pb(\eta'>\xi')+0.5\Pb(\eta'=\xi').
\end{align}
The {\it stratified win probability} is defined as 
\begin{align*}
\theta^{str}=\omega \theta+(1-\omega )\theta',\ \ w\in(0,1).
\end{align*}
The null hypothesis for the stratified analysis is 
\begin{align}\label{Eq32}
&\H_0: \ \ \theta=\frac{1}{2}\text{ and } \theta'=\frac{1}{2}\text{ against}\nonumber\\
&\H_1: \ \ \theta\neq\frac{1}{2}\text{ or } \theta'\neq\frac{1}{2}.
\end{align}
Under the null hypothesis $\theta^{str}=\theta=\theta'=\frac{1}{2}$ regardless of the value of $\omega $. Sometimes instead of the weights we will specify only the coefficients $\omega _1,\omega _2$ per stratum, and the weights can be calculated using the formula 
$$\omega =\frac{\omega _1}{\omega _1+\omega _2},\ \ 1-\omega =\frac{\omega _2}{\omega _1+\omega _2}.$$

For each stratum separately, we can construct the win proportions \eqref{Eq5}, denoted correspondingly by $\hat\theta_N$ and $\hat\theta_N'.$  A general method of combining these estimates is to use the weighted sum of the estimates in each stratum as
\begin{align}\label{Eq42}
\hat\theta_N^{str}=\omega \hat\theta_N+(1-\omega)\hat\theta_N'.
\end{align}
The variance of the stratified estimator can be calculated using the formula (since observations in different stratum are independent)
\begin{align}\label{Eq43}
\V ar(\hat\theta_N^{str})=\omega^2\V ar(\hat\theta_N)+\left(1-\omega \right)^2\V ar(\hat\theta_N').
\end{align}
The variances of the win proportions inside each stratum can be estimated using Theorems \ref{T1} and \ref{T3} as
$$\V ar(\hat\theta_N)=\hat\sigma_N^2 \text{ and } \V ar(\hat\theta_N')=(\hat\sigma_N')^2.$$
The weights can be estimated (see Appendix II in \cite{Koch1998}) using the coefficients
$$w_1=\frac{n_{11}n_{12}}{n_1},\, w_2=\frac{n_{21}n_{22}}{n_2},$$
which will give the following estimates of the weights:
\begin{align}\label{Eq44}
 w=\frac{\frac{1}{n_{21}}+\frac{1}{n_{22}}}{\frac{1}{n_{11}}+\frac{1}{n_{12}}+\frac{1}{n_{21}}+\frac{1}{n_{22}}},\ \ 1-w=\frac{\frac{1}{n_{11}}+\frac{1}{n_{12}}}{\frac{1}{n_{11}}+\frac{1}{n_{12}}+\frac{1}{n_{21}}+\frac{1}{n_{22}}}.
\end{align}

\begin{remark}\label{R11} If a balanced design between the treatment groups and the strata is applicable, that is, $n_{11}=n_{12}=n_{21}=n_{22},$ then both weights would be $0.5$. If the treatment group has the same proportion in both strata, $\frac{n_{11}}{n_{12}}=\frac{n_{21}}{n_{22}}=\frac{n_{1}}{n_{2}},$  that is, only balanced allocation of treatment within a stratum is present, then
$$ w=\frac{n_1}{n_1+n_2},$$
and so the larger stratum will get bigger weight.
\end{remark}
Another possible choice of coefficients (see Section \ref{CMH}) is
\begin{align}\label{Eq45}
w_1^0=\frac{n_{11}n_{12}}{n_1+1},\ \ w_2^0=\frac{n_{21}n_{22}}{n_2+1},
\end{align}
which will give the {\it van Elteren weight}, denoted by $w^0$. 
\begin{theorem}\label{T10} For the weights \eqref{Eq44}, \eqref{Eq45} the following convergence holds
$$Z_N^{str}=\frac{\hat\theta_N^{str}-\theta^{str}}{\sqrt{\V ar(\hat\theta_N^{str})}}\Longrightarrow\N(0,1),\ \ n_{ij}\rightarrow+\infty,\,i=1,2,\,j=1,2.$$
\end{theorem}

\begin{remark}\label{R12} The theorem above is true for a large family of weights depending only on the sample size.
\end{remark}

\subsection{Adjusted win probability with stratification}\label{AWPS}
For a randomized trial, the null hypothesis for the adjusted win probability with stratification is 
\begin{align}\label{Eq53}
&\H_0: \ \ \beta=\frac{1}{2}\text{ and } \beta'=\frac{1}{2}\text{ against}\nonumber\\
&\H_1: \ \ \beta\neq\frac{1}{2}\text{ or } \beta'\neq\frac{1}{2},
\end{align}
where $\beta,\beta'$ are the adjusted win probabilities per stratum. The observed sample consists of response vectors and numeric covariate vectors for two treatment groups (as in \eqref{Eq11}), observed independently for both strata, denoted
\begin{align}\label{Eq47}
\text{Stratum I } (Y_1,X_1),\ \ (Y_2,X_2),\nonumber\\
\text{Stratum II } (Y_1',X_1'),\ \ (Y_2',X_2').
\end{align}
Here, as usual, the first group is the placebo group and the second group is the active group. As in the previous sections we need to replace the ordinal response variables with individual win proportions. 
\begin{align*}
\text{Stratum I } (Y_1^0,X_1),\ \ (Y_2^0,X_2),\nonumber\\
\text{Stratum II } ((Y_1')^0,X_1'),\ \ ((Y_2')^0,X_2').
\end{align*}
The individual win proportions are calculated for each stratum separately, disregarding the values of covariates. There are several methods to estimate the adjusted win probability with stratification. One way of achieving this is to make covariate adjustment first, then use weights to combine these estimators. The approach used in this section will follow \cite{Koch1998}, where stratification is performed first, separately for the crude win proportions and the covariates using the same weights, and then the adjustment of stratified win proportion with stratified covariate is made. Following Section \ref{SWP} we can construct the stratified win probability $\hat\theta_N^{str}$ (see \eqref{Eq42}). On the other hand, the mean difference of covariates (see \eqref{Eq16}) will serve as the statistic for comparison of covariates between treatment groups
\begin{align}\label{Eq48}
\bar x=\bar x_1-\bar x_2,\ \ \bar x'=\bar x_1'-\bar x_2'. 
\end{align}
Using the same weights as for the construction of $\hat\theta_N^{str},$ we can construct the stratified mean difference of covariates
\begin{align}\label{Eq49}
\bar x_N^{str}= w\bar x+(1-w)\bar x'.
\end{align}
Therefore, an estimate for the adjusted win probability with stratification is
\begin{align}\label{Eq50}
\hat\beta_N^{str}=\hat\theta_N^{str} - \frac{\bar x_N^{str}}{\V ar(\bar x_{N}^{str})}\C ov(\bar x_{N}^{str},\hat\theta_N^{str}).
\end{align}
Here $\V ar(\bar x_{N}^{str})$ and $\C ov(\bar x_{N}^{str},\hat\theta_N^{str})$ are calculated as $\V ar(\hat\theta_N^{str})$ in \eqref{Eq43}
\begin{align}\label{Eq51}
\V ar(\bar x_{N}^{str})= w^2\V ar(\bar x_{N})+\left(1- w\right)^2\V ar(\bar x_{N}'),\nonumber\\
\C ov(\bar x_{N}^{str},\hat\theta_N^{str})= w^2\C ov(\bar x_{N},\hat\theta_N)+\left(1- w\right)^2\C ov(\bar x_{N}',\hat\theta_N'),
\end{align}
while inside each stratum  $\V ar(\bar x_{N}),\,\C ov(\bar x_{N},\,\hat\theta_N),\,\V ar(\hat\theta_N)$ are calculated as in \eqref{Eq16},\eqref{Eq17}. Below are the formulas for the stratum I
\begin{align}\label{Eq52}
&\V ar(\hat\theta_N)=\frac{\V ar(y_1^0)}{n_{11}} + \frac{\V ar(y_2^0)}{n_{12}},\ \ \V ar(\bar x_{N})=\frac{\V ar(x_1)}{n_{11}} + \frac{\V ar(x_2)}{n_{12}},\nonumber\\
&\C ov(\bar x_{N},\,\hat\theta_N)=\frac{\C ov(x_1,y_1^0)}{n_{11}} + \frac{\C ov(x_2,y_2^0)}{n_{12}}.
\end{align}
The following result is from \cite{Koch1998}.
\begin{theorem}\label{T12} The following asymptotic result holds for a randomized trial
$$Z_N^{AdS}=\frac{\hat\beta_N^{str}-\beta^{str}}{\sqrt{\V ar(\hat\beta_N^{str})}}\Longrightarrow\N(0,1),\ \ n_{ij}\rightarrow+\infty,\,i=1,2,\,j=1,2,$$
where $\beta^{str}=w\beta+(1-w)\beta',\ \ w\in(0,1)$ and 
$$\V ar(\hat\beta_N^{str})=\V ar(\hat\theta_N^{str})-\frac{[\C ov(\bar x_{N}^{str},\,\hat\theta_N^{str})]^2}{ \V ar(\bar x_{N}^{str})}.$$
\end{theorem}

\section{Win ratio and rank tests}\label{III}
In this section we consider several well known tests and compare them with the tests described in previous sections. The test for non-adjusted win ratio (see the Section \ref{WR}) will be compared with the tests for the location problem ({\it Wilcoxon two sample rank-sum test} and the {\it Fligner-Policello test}) and the win probability itself will be compared with the {\it Hodges-Lehmann estimator}. The test for the stratified win probability (see the Section \ref{SWP}) will be compared with the {\it Cochran-Mantel-Haenszel test} and the test based on the adjusted win probability with stratification (see the Section \ref{AWPS}) will be compared with the {\it rank ANCOVA}. 

\subsection{Tests for the location problem}\label{TLP}

Testing the hypothesis $\theta=\frac{1}{2}$ or, equivalently, $\kappa=1$ is closely related to location testing in a situation where the distribution function of the second sample represents a shifted version of the distribution of the first sample (see Remark \ref{R7}). Given the two samples \eqref{Eq1}, suppose that their distribution functions differ only by a shift
\begin{align*}
F_\xi(y)=\Pb(\xi\leq y),\ \ F_\eta(y)=\Pb(\eta\leq y)=F_\xi(y-\delta), \text{ for all } y\in\R.
\end{align*} 
The hypothesis of interest is
\begin{align}\label{Eq24}
\H_0: \ \ \delta=0\,  \text{ against } \H_1: \ \ \delta\neq0.
\end{align}
In the subsequent sections  we will compare the estimators and the tests for the win probability with the shift estimators and tests in the location problem.
\subsubsection{Wilcoxon two-sample rank-sum test}\label{WTST}
The following theorem (see, for example, \cite{Gib2010}, page 290) allows construction of asymptotic tests for the hypothesis \eqref{Eq24}.

\begin{theorem}[Wilcoxon]\label{T2}
Under the null hypothesis $\delta=0$ the following convergence for the Wilcoxon rank-sum statistic $W= \sum_{j=1}^{n_2}R_{2j}$ holds 
\begin{align*}
Z_N=\frac{W-n_2\bar R_N}{\sqrt{\frac{n_1n_2}{N}\V ar(R)}}\Longrightarrow\N(0,1),\ \ \text{as }N\rightarrow+\infty,
\end{align*}
where the ranks of the sample \eqref{Eq18} are denoted by $R=(R_{11},\cdots,R_{1n_1},R_{21},\cdots,R_{2n_2}).$ The mean and the variance of all ranks are denoted correspondingly by $\bar R_N=\frac{1}{N}(\sum_{i=1}^{n_1}R_{1i}+\sum_{j=1}^{n_2}R_{2j})$ and $$\V ar(R)=\frac{1}{N-1}\sum_{i=1}^{n_1}\left(R_{1i}-\bar R_N\right)^2+\frac{1}{N-1}\sum_{j=1}^{n_2}\left(R_{2j}-\bar R_N\right)^2.$$
\end{theorem}
Modifying the $Z$ value from the Theorem \ref{T2} using the equality (see \eqref{Eq20})
\begin{align*}
\sum_{i=1}^{n_1}R_{1i}+\sum_{j=1}^{n_2}R_{2j}=\frac{N(N+1)}{2}, 
\end{align*}
we can write
\begin{align}\label{Eq25}
Z_N=\frac{W-n_2\frac{N+1}{2}}{\sqrt{\frac{n_1n_2}{N}\V ar(R)}}.
\end{align}
In the location problem the hypothesis $\delta=0$ is equivalent to $\theta=\frac{1}{2},$ hence, from Theorem \ref{T3} we can construct an asymptotic test based on the following statistic
\begin{align}\label{Eq26}
\tilde Z_N=\frac{\hat\theta_N-\frac{1}{2}}{\hat\sigma_N}.
\end{align}
\begin{remark}\label{R4} 
For the {\it Wilcoxon rank-sum test}, the estimated variance in Theorem \ref{T2} is only applicable under the null hypothesis where the distributions of response variables for the two groups are identical (under the null hypothesis in the specification \eqref{Eq24}). More generally, the estimated variance shown in Theorem \ref{T1}, and in its equivalent formulation in Theorem \ref{T3}, is applicable regardless of whether the two groups have the same distribution and thereby is applicable more broadly than under a null hypothesis for either $\theta=\frac{1}{2}$ or equality of distributions. 
\end{remark}
Theorem \ref{T4} compares the {\it Wilcoxon rank-sum test} and the test based on the win probability under the null hypothesis of the location testing specification \eqref{Eq24} (equality of distributions). 
\begin{theorem}\label{T4}
Under the null hypothesis of \eqref{Eq24}, that is when the distributions of random variables $\xi$ and $\eta$ are equal, the following relationship holds  for the statistics \eqref{Eq25}, \eqref{Eq26} 
\begin{align*}
\frac{\tilde Z_N^2}{Z_N^2}=\frac{\V ar(R)}{Nn_1n_2\hat\sigma_N^2}=\frac{1}{n_1n_2}\frac{\frac{\V ar(R)}{N}}{\frac{\V ar(Y_2^0)}{n_2}+\frac{\V ar(Y_1^0)}{n_1}}.
\end{align*}
\end{theorem}
\begin{proof}
The proof directly follows from Theorems \ref{T2} and \ref{T3}, since from \eqref{Eq20} we have
\begin{align*}
\hat \theta_N-\frac{1}{2}=\frac{W}{n_1n_2}-\frac{n_2(n_2+1)}{2n_1n_2}-\frac{1}{2}=\frac{1}{n_1n_2}\left(W-\frac{n_2(N+1))}{2}\right).
\end{align*}
Hence from \eqref{Eq26} we get
\begin{align*}
\tilde Z_N=\frac{1}{n_1n_2\hat\sigma_N}\left(W-\frac{n_2(N+1))}{2}\right).
\end{align*}
Comparing the last expression with \eqref{Eq25} we will get the proof of the theorem.
\end{proof}
If $n_1=n_2=\frac{N}{2}$ then 
\begin{align*}
\frac{\tilde Z_N^2}{Z_N^2}=\frac{2}{N^2}\frac{\V ar(R)}{\V ar(Y_2^0)+\V ar(Y_1^0)}.
\end{align*}
Under the null hypothesis $\delta=0$ we have equality of distributions of random variables $\xi$ and $\eta$, and if additionally we suppose that this distribution is continuous then (see, for example, page 14, \cite{Leh1975}) $$\V ar(R)=\frac{N(N+1)}{12},$$
then Theorem \ref{T4} can be simplified as follows
\begin{align*}
\frac{\tilde Z_N^2}{Z_N^2}=\frac{1}{3}\frac{N+1}{N^2}\frac{1}{\hat\sigma_N^2}=\frac{N+1}{6N}\frac{1}{\V ar(Y_2^0)+\V ar(Y_1^0)}.
\end{align*}
\begin{remark}\label{R9} 
The quantity $\frac{\tilde Z_N^2}{Z_N^2}$ can be greater or less than 1. Hence it is not possible to say in advance which test will give smaller p-value.
\end{remark}

\subsubsection{The Hodges-Lehmann estimator of location shift}\label{HLE}
The previous section described a method of detecting a significant shift for the location testing problem \eqref{Eq24}. Here we will discuss the Hodges-Lehmann estimator for the shift $\delta,$ which can be used to quantify the treatment effect, while presenting the significance testing based on Theorem \ref{T2}. In this section we additionally require that the random variables $\xi$ and $\eta$ be numeric, to allow calculations on the samples \eqref{Eq1}. 

Consider all possible differences between the two groups in the sample \eqref{Eq1}
\begin{align*}
D_{ji}=Y_{2j}-Y_{1i},\ \ j=1,\cdots,n_2,\ \ i=1,\cdots,n_1.
\end{align*}
The median of all these $n_1n_2$ differences is called the {\it Hodges-Lehmann} estimator (see, for example, \cite{Leh1975}). Hence if we order all the values $D_{ji}$ increasingly and denote the ordered sample by $D_{(k)},\,k=1,\cdots,n_1n_2$ then the {\it Hodges-Lehmann} estimator is equal to
\begin{align*}
\hat \delta&=D_{(\tilde k)},\, \tilde k=\frac{n_1n_2+1}{2},\,n_1n_2 \text{ if is odd,}\\
\hat \delta&=\frac{D_{(\tilde k)}+D_{(\tilde k+1)}}{2},\,\tilde k=\frac{n_1n_2}{2}, \text{ otherwise.}
\end{align*}
It is well known that this estimator does not work well in the presence of ties. Consider a simple model $Y_2=(1,1,2)$ and $Y_1=(0,1,2).$ The ordered differences will be $(-1,-1,0,0,0,1,1,1,2)$ hence $\hat \delta=0$. Which means that the Hodges-Lehmann estimator does not detect the difference between these two samples. However, if we calculate the win proportion of the second group against the first group we will get $\hat\theta_N=0.61.$ Hence, the win proportion is a more adequate measure of difference in this sample. It also provides a corresponding hypothesis test result and requires fewer assumptions. Although there are methods of constructing confidence intervals for the Hodges-Lehmann estimator (see, for example, \cite{Hol1999}), the hypothesis test is usually done using the Wilcoxon test and is not based on the confidence interval itself which is an additional disadvantage of the Hodges-Lehmann estimator. For the Hodges-Lehmann estimator to have good properties, additional assumptions of $\eta-\xi$ having a continuous, symmetric distribution are needed  (\cite{Hol1999}).

\subsubsection{The Fligner-Policello test}\label{FPT}

The Fligner-Policello test (\cite{Flig1981}) is similar to the test based on Theorem \ref{T1}. It uses quantities called {\it placements} which are the same as the individual win proportions (see \eqref{Eq6}, \eqref{Eq15}) and the statistic for testing the hypothesis is based on the samples \eqref{Eq7}.
\begin{theorem}[Fligner-Policello]\label{T5}
In the location testing problem suppose that $F_\xi(\cdot)$ and $F_\eta(\cdot)$ are different. In additional, require these distributions to be symmetric. Denote by $\mu_\xi$ and $\mu_\eta$ the medians (assumed unique) of distributions $\xi$ and $\eta$, respectively. Then, to test the hypothesis
\begin{align*}
\H_0: \ \  \mu_\xi=\mu_\eta\,  \text{ against } \H_1: \ \ \mu_\xi\neq\mu_\eta,
\end{align*}
the following statistic can be used
\begin{align}\label{Eq27}
F_N=\frac{n_1n_2\hat\theta_N-n_1n_2(1-\hat\theta_N)}{2\sqrt{n_2^2(n_1-1)\V ar(Y_1^0)+n_1^2(n_2-1)\V ar(Y_2^0)+n_1n_2\hat\theta_N(1-\hat\theta_N)}},
\end{align}
where $\hat\theta_N$ is the win proportion (see \eqref{Eq9}). Under the null hypothesis $$F_N\Longrightarrow\N(0,1),\text{ as } n_1\rightarrow+\infty,n_2\rightarrow+\infty.$$
\end{theorem}

\begin{remark}\label{R2} 
Assumptions of uniqueness of medians and the symmetry of distributions is only needed to have the equivalence of conditions $\mu_\eta-\mu_\xi>0$ and $\theta>\frac{1}{2}.$ If we test the hypothesis not on the medians but directly on the win probability $\theta$, then these conditions are not needed (as in Theorem \ref{T1}).  
\end{remark}
From Theorems \ref{T2}, \ref{T3} we have
\begin{align*}
Z_N=\frac{\hat\theta_N-\theta}{\sqrt{\frac{\V ar(Y_2^0)}{n_2}+\frac{\V ar(Y_1^0)}{n_1}}}=\frac{\hat\theta_N-\theta}{\hat\sigma_N^2},
\end{align*}
where the variances are estimated using the consistent estimators
\begin{align*}
\V ar(Y_2^0)=\frac{1}{n_2-1}\sum_{j=1}^{n_2}(p_j-\hat\theta_n)^2,\ \ \V ar(Y_1^0)=\frac{1}{n_1-1}\sum_{i=1}^{n_1}(q_i-(1-\hat\theta_n))^2.
\end{align*}
Evidently if we consider the statistic where variances are estimated with the coefficients $1/n$ instead of $1/(n-1)$, 
\begin{align}\label{Eq28}
Z_N^0=\frac{\hat\theta_N-\theta}{\sqrt{\frac{1}{n_2^2}\sum_{j=1}^{n_2}(p_j-\hat\theta_n)^2+\frac{1}{n_1^2}\sum_{i=1}^{n_1}(q_i-(1-\hat\theta_n))^2}},
\end{align}
the asymptotic result of Theorem \ref{T1} will still be true.
\begin{theorem}\label{T6}
Under the null hypothesis $\theta=\frac{1}{2}$, for the Fligner-Policello statistic $F_N$ and the statistic $Z_N^0$ defined in \eqref{Eq28}, the following relationship holds 
\begin{align*}
\frac{(Z_N^0)^2}{F_N^2}\geq 1.
\end{align*}
\end{theorem}
\begin{proof}
The proof immediately follows from \eqref{Eq27}. Indeed,
\begin{align*}
F_N=\frac{\hat\theta_N-\frac{1}{2}}{\sqrt{\frac{1}{n_2^2}\sum_{j=1}^{n_2}(p_j-\hat\theta_n)^2+\frac{1}{n_1^2}\sum_{i=1}^{n_1}(q_i-(1-\hat\theta_n))^2+\frac{1}{n_1n_2}\hat\theta_N(1-\hat\theta_N)}},
\end{align*}
therefore $\frac{(Z_N^0)^2}{F_N^2}\geq 1$ if $\theta=\frac{1}{2}.$
\end{proof}
\begin{remark}\label{R6} 
The result of Theorem \ref{T6} means that the test based on the statistic $Z_N^0$ for each value of the sample size $N$ will give smaller p-value, than the p-value based on the Fligner-Policello test. Thus, the former needs less evidence to reject the null hypothesis than the Fligner-Policello test. For large values of the sample size the tests based on $Z_N,\,Z_N^0$ and the Fligner-Policello test will give similar results, since the difference in standard errors is $\frac{1}{n_1n_2}\hat\theta_N(1-\hat\theta_N),$ which tends to 0 as $n_1\rightarrow+\infty,\,n_2\rightarrow+\infty.$
\end{remark}

\subsection{Regression on ranks}\label{RR}
With respect to Section \ref{AWP}, in this section we will consider the situation when a numeric covariate is present and the group comparison of ordinal response variables should be adjusted for that covariate. \cite{Koch1982} describes one way of conducting an adjusted comparison of groups by using {\it the regression on ranks.} The observed sample is \eqref{Eq11}. The vector of response values $Y=(Y_1,Y_2)$ is replaced by the vector of ranks $$R=(R_{11},\cdots,R_{1n_1},R_{21},\cdots,R_{2n_2}).$$
In the first step a simple linear regression is fitted for the pair $(R,X),$ where $X=(x_{11},\cdots,x_{1n_1},x_{21},\cdots,x_{2n_2})$ is the vector of covariates for the combined two treatment groups. Denoting by $\bar R,\bar X$ correspondingly the means of the ranks and the covariates, the estimates of the intercept and the slope of this linear regression would be
\begin{align*}
\hat\psi_0=\bar R-\hat\psi_1\bar X,\ \ \hat\psi_1=\frac{\C ov(R,X)}{\V ar(X)}. 
\end{align*}
In the next step we replace the ranks by their residuals after fitting the linear regression described above
\begin{align*}
R_{ij}^{res}=R_{ij}-(\hat\psi_0+\hat\psi_1x_{ij}).
\end{align*}
The formula for the residuals can be simplified as follows
\begin{align}\label{Eq40}
 R_{ij}^{res}=(R_{ij}-\bar R)-(x_{ij}-\bar X)\frac{\C ov(R,X)}{\V ar(X)}.
\end{align}
Then the computations for the Wilcoxon test (see Theorem \ref{T2}) are applied to the residuals in \eqref{Eq40}, and this means that the $Z$ statistic in \eqref{Eq25} is constructed based on the residuals $R_{ij}^{res}$. On the other hand, if we calculate the sum of residual ranks of the active group, then
\begin{align*}
 \sum_{j=1}^{n_2}R_{2j}^{res}=\sum_{j=1}^{n_2}R_{2j}-n_2\bar R-\frac{n_1n_2}{N}(\bar x_{2}-\bar x_{1})\frac{\C ov(R,X)}{\V ar(X)},
\end{align*}
hence, remembering also that $\hat\theta_N-\frac{1}{2}=\frac{1}{n_1n_2}\left(\sum_{j=1}^{n_2}R_{2j}-n_2\bar R\right),$ we get
\begin{align}\label{Eq60}
\sum_{j=1}^{n_2}R_{2j}^{res}=n_1n_2\left(\hat\theta_N-\frac{1}{2}\right)-\frac{n_1n_2}{N}(\bar x_{2}-\bar x_{1})\frac{\C ov(R,X)}{\V ar(X)}.
\end{align}
It is easy to see that
\begin{align*}
\V ar(R^{res})=\V ar(R)-\frac{\C ov(R,X)^2}{\V ar(X)}.
\end{align*}
Therefore the $Z$ value based on residuals will be (remembering that the sum of residuals is always equal to zero)
\begin{align}\label{Eq41}
 Z_N^{res}=\frac{\sum_{j=1}^{n_2}R_{2j}^{res}}{\sqrt{\frac{n_1n_2}{N}\V ar(R^{res})}}=\frac{\hat\theta_N-\frac{1}{2}-\frac{1}{N}(\bar x_{2}-\bar x_{1})\frac{\C ov(R,X)}{\V ar(X)}}{\sqrt{\frac{1}{n_1n_2N}\left[\V ar(R)-\frac{\C ov(R,X)^2}{\V ar(X)}\right]}}.
\end{align}
Or, using the equation
\begin{align}\label{Eq58}
\hat\theta_N-\frac{1}{2}=\frac{1}{n_1n_2}\left(\sum_{j=1}^{n_2}R_{2j}-n_2\bar R\right)=\frac{R_{2n_2}-R_{1n_1}}{N},
\end{align}
we can rewrite \eqref{Eq41} as
\begin{align}\label{Eq59}
 Z_N^{res}=\frac{1}{N}\frac{R_{2n_2}-R_{1n_1}-(\bar x_{2}-\bar x_{1})\frac{\C ov(R,X)}{\V ar(X)}}{\sqrt{\frac{1}{n_1n_2N}\left[\V ar(R)-\frac{\C ov(R,X)^2}{\V ar(X)}\right]}}.
\end{align}
Equations \eqref{Eq41} and \eqref{Eq59} can be compared with the $Z$ value under the null hypothesis $\theta=\frac{1}{2}$ from Theorem \ref{T7}
\begin{align*}
Z_N^\beta=\frac{\hat\theta_N-\frac{1}{2}-\frac{\bar x_1- \bar x_2}{\frac{\V ar(x_{1})}{n_1}+\frac{\V ar(x_{2})}{n_2}}\left[\frac{\C ov(x_1,y_1^0)}{n_1} + \frac{\C ov(x_2,y_2^0)}{n_2}\right]}{\sqrt{\frac{\V ar(Y_2^0)}{n_2}+\frac{\V ar(Y_1^0)}{n_1}- \frac{\left[\frac{\C ov(x_1,y_1^0)}{n_1} + \frac{\C ov(x_2,y_2^0)}{n_2}\right]^2}{\frac{\V ar(x_{1})}{n_1}+\frac{\V ar(x_{2})}{n_2}}}}.
\end{align*}
\begin{remark}\label{R10} The comparison of standard errors of win proportion and the regression on the ranks is similar to the case of crude (non-adjusted) estimates in Theorem \ref{T4}; only a coefficient $\frac{1}{n_1n_2}$ is added. Formula \eqref{Eq58}, like formula \eqref{Eq23}, is another computationally less intense method to calculate the win proportion.
\end{remark}

\subsection{The Cochran-Mantel-Haenszel test}\label{CMH}

In was shown in Section \ref{CWS} that there is equivalence between the Mann-Whitney U-statistic and the Wilcoxon rank-sum statistic. Both statistics can be generalized to stratified analysis. For stratified analysis, the generalization of the Mann-Whitney test is called the {\it Cochran-Mantel-Haenszel test}, and the generalization of the Wilcoxon rank-sum test is called the {\it Van Elteren test}. 

Consider the problem of stratified analysis \eqref{Eq30}, \eqref{Eq31}. Unlike the null hypothesis \eqref{Eq32}, a stronger condition will be subject to testing, that is, the equality of distributions in each stratum,
\begin{align}\label{Eq33}
&\H_0: \ \ F_\xi(\cdot)=F_\eta(\cdot)\text{ and }  F_{\xi'}(\cdot)=F_{\eta'}(\cdot),\text{ against}\nonumber\\
&\H_1: \ \ F_\xi(\cdot)\neq F_\eta(\cdot)\text{ or }  F_{\xi'}(\cdot)\neq F_{\eta'}(\cdot).
\end{align}
As in Section \ref{WTST} we can construct the win proportion and the corresponding $Z$ statistic (see Theorem \ref{T2}) for each stratum separately. Denote by $R$ the vector of ranks of the sample $Y=(Y_1,Y_2),$ and by $R'$ the vector of ranks of the sample $Y'=(Y_1',Y_2').$ Then
\begin{align}\label{Eq36}
Z_{n_1}=\frac{W-n_{12}\bar R_{n_1}}{\sqrt{\frac{n_{11}n_{12}}{n_1}\V ar(R)}},\ \ Z_{n_2}'=\frac{W'-n_{22}\bar R_{n_2}'}{\sqrt{\frac{n_{21}n_{22}}{n_2}\V ar(R')}}.
\end{align}
The idea for construction of the stratified statistic is to choose coefficients per strata $(w_1,w_2)$ and combine the estimates for each stratum using these coefficients
\begin{align*}
W^{str}=w_1W+w_2W'. 
\end{align*}
The $Z$ statistic for the hypothesis testing can be obtained by combining the $Z$ values in \eqref{Eq36} using the same coefficients
\begin{align*}
Z_N^{Elt}=\frac{W^{str}-(w_1n_{12}\bar R_{n_1}+w_2n_{22}\bar R_{n_2}')}{\sqrt{w_1^2\frac{n_{11}n_{12}}{n_1}\V ar(R)+w_2^2\frac{n_{21}n_{22}}{n_2}\V ar(R')}}.
\end{align*}
The {\it van Elteren statistic} (see, for example, page 145, \cite{Leh1975}) for testing the hypothesis in \eqref{Eq33} uses the following coefficients 
\begin{align*}
w_1=\frac{1}{n_1+1},\ \ w_2=\frac{1}{n_2+1}. 
\end{align*}
Comparing to the van Elteren coefficients \eqref{Eq45} that were used to estimate the stratified win probability we get
\begin{align}\label{Eq39}
n_{11}n_{12}w_1=w_1^0,\ \ n_{21}n_{22}w_2=w_2^0. 
\end{align}
 
\begin{theorem}[van Elteren]\label{T9}
Under the null hypothesis \eqref{Eq33} and for the coefficients defined in \eqref{Eq39} the following convergence holds 
$$Z_N^{Elt}\Longrightarrow\N(0,1),\ \ n_{ij}\rightarrow+\infty,\,i=1,2,\,j=1,2.$$
\end{theorem}
For the relationship of the van Elteren test and the Cochran-Mantel-Haenszel test see \cite{Land1978}. Here again, as in Theorem \ref{T4} we can compare the statistic from this test with the statistic based on the win probability described in Section \ref{SWP}. Using the equalities 
\begin{align*}
\hat\theta_N-\frac{1}{2}=\frac{1}{n_{11}n_{12}}(W-n_{12}\bar R_N),\ \ \hat\theta_N'-\frac{1}{2}=\frac{1}{n_{21}n_{22}}(W'-n_{22}\bar R_N'),
\end{align*}
we get the following result
\begin{theorem}\label{T11}  For the $Z_N^{Elt}$ statistic from the van Elteren test and the $Z_N^{str}$ statistic based on the stratified win probability with van Elteren weights, under the null hypothesis of equality of distributions in each stratum, we have
$$\frac{(Z_N^{str})^2}{(Z_N^{Elt})^2}=\frac{\frac{1}{n_{11}n_{12}}(w_1^0)^2\frac{\V ar(R)}{n_1}+\frac{1}{n_{21}n_{22}}(w_2^0)^2\frac{\V ar(R')}{n_2}}{\left(w_1^0\right)^2\V ar(\hat\theta_N)+\left(w_2^0\right)^2\V ar(\hat\theta_N')}.$$
Here $w_1^0,w_2^0$ are the van Elteren coefficients defined in \eqref{Eq45}.
\end{theorem}

\begin{remark}\label{R17} The comparison of $Z$ statistics follows the same pattern as in Remark \ref{R10}, that is coefficients $\frac{1}{n_{11}n_{12}}$ and $\frac{1}{n_{21}n_{22}}$ are added per stratum.
\end{remark}

\subsection{The rank ANCOVA}\label{RANK}

In this section we will compare the statistical hypothesis test based on the adjusted and stratified win probability described in Section \ref{AWPS} with the rank ANCOVA test. The rank ANCOVA was proposed in \cite{Quade1967} (see also {\it Section 7.6, Rank Analysis of Covariance} in \cite{Stokes2012}, as well as \cite{Koch1982},\cite{Koch1990}). The test based on the rank ANCOVA approach is very similar to the test obtained from the regression on ranks in Section \ref{RR}. Here again, as the first step, a simple regression line will be fitted for the ranks, but in the second step, instead of the Wilcoxon test, the van Elteren test is performed, to account for the stratification as well.

The hypothesis to test is \eqref{Eq53} while observing the samples \eqref{Eq47}. Introducing the ranks $R,R'$ correspondingly for the combined samples $Y=(Y_1,Y_2)$ and $Y'=(Y_1',Y_2')$. As in Section \ref{RR} we fit to the ranks $R,R'$ their respective (combined across treatment groups) covariate vectors $X,X'$ and derive (see \eqref{Eq40}) the residuals $R^{res},(R^{res})'$ per each stratum. The van Elteren test, described in the Section \ref{CMH}, can be applied to the residuals to test for treatment effect difference across strata. Remembering that the sum of all residuals in a stratum is zero, the van Elteren statistic will be
\begin{align}\label{Eq54}
Z_N^{A}=\frac{w_1W^{res}+w_2(W^{res})'}{\sqrt{w_1^2\frac{n_{11}n_{12}}{n_1}\V ar(R^{res})+w_2^2\frac{n_{21}n_{22}}{n_2}\V ar((R^{res})')}}.
\end{align}
Here $w_1=\frac{1}{n_1+1}$ and $w_2=\frac{1}{n_2+1}$. The test based on this $Z_N^{A}$ value will be called rank ANCOVA test. The following result is from \cite{Quade1967}
\begin{theorem}\label{T13} The following asymptotic result holds
$$Z_N^{A}\Longrightarrow\N(0,1),\ \ n_{ij}\rightarrow+\infty,\,i=1,2,\,j=1,2.$$
\end{theorem}
Here we can draw parallels between the rank ANCOVA test and the test based on the adjusted win probability with stratification described in Section \ref{AWPS}. For each stratum separately we can write the formula \eqref{Eq60}
\begin{align*}
&W^{res}=\sum_{j=1}^{n_{12}}R_{2j}^{res}=n_{11}n_{12}\left(\hat\theta_{n_1}-\frac{1}{2}\right)-\frac{n_{11}n_{12}}{n_1}(\bar x_{2}-\bar x_{1})\frac{\C ov(R,X)}{\V ar(X)},\\
&(W^{res})'=\sum_{j=1}^{n_{22}}(R_{2j}^{res})'=n_{21}n_{22}\left(\hat\theta_{n_2}'-\frac{1}{2}\right)-\frac{n_{21}n_{22}}{n_2}(\bar x_{2}'-\bar x_{1}')\frac{\C ov(R',X')}{\V ar(X')}.
\end{align*}
Using the notation \eqref{Eq39} we set
\begin{align*}
&w_1W^{res}=w_1^0\left(\hat\theta_{n_1}-\frac{1}{2}\right)-\frac{w_1^0}{n_1}(\bar x_{2}-\bar x_{1})\frac{\C ov(R,X)}{\V ar(X)},\\
&w_2(W^{res})'=w_2^0\left(\hat\theta_{n_2}'-\frac{1}{2}\right)-\frac{w_2^0}{n_2}(\bar x_{2}'-\bar x_{1}')\frac{\C ov(R',X')}{\V ar(X')}.
\end{align*}
Hence, dividing also the numerator and the denomination by $(w_1^0+w_2^0)$ and remembering that $w^0=\frac{w_1^0}{w_1^0+w_2^0}$ is the van Elteren weight, we can simplify \eqref{Eq54} as follows
\begin{align}\label{Eq56}
 Z_N^{A}=\frac{\hat\theta_N^{str}-\frac{1}{2}-\left[\frac{w^0}{n_1}(\bar x_{2}-\bar x_{1})\frac{\C ov(R,X)}{\V ar(X)}+\frac{(1-w^0)}{n_2}(\bar x_{2}'-\bar x_{1}')\frac{\C ov(R',X')}{\V ar(X')}\right]}{\sqrt{\frac{(w^0)^2}{n_1}\left[\V ar(R)-\frac{\C ov(R,X)^2}{\V ar(X)}\right]+\frac{(1-w^0)^2}{n_2}\left[\V ar(R')-\frac{\C ov(R',X')^2}{\V ar(X')}\right]}},
\end{align}
here $\hat\theta_N^{str}=w^0\hat\theta_{n_1}+(1-w^0)\hat\theta_{n_2}'.$ On the other hand, under the null hypothesis of adjusted win probability with stratification being $\beta=\frac{1}{2},$ we derive from Theorem \ref{T12} (see formulas \eqref{Eq48}-\eqref{Eq52})
\begin{align}\label{Eq61}
Z_N^{AdS}=\frac{\hat\theta_N^{str} - \frac{1}{2}-\frac{\bar x_N^{str}}{\V ar(\bar x_{N}^{str})}\C ov(\bar x_{N}^{str},\hat\theta_N^{str})}{\sqrt{\V ar(\hat\theta_N^{str})-\frac{[\C ov(\bar x_{N}^{str},\,\hat\theta_N^{str})]^2}{ \V ar(\bar x_{N}^{str})}}}.
\end{align}

\begin{remark}\label{R13} To compare the $Z$ values from \eqref{Eq56} and \eqref{Eq61}, we see the following differences. First, the rank ANCOVA approach uses combined estimates of variances across treatment groups in each stratum, while the test based on the adjusted win probability with stratification uses pooled estimates of variances. Second, the rank ANCOVA approach performs adjustment by the numeric covariate first, then combines estimates across strata, whereas in the win probability approach the non-adjusted win proportion, the covariates are combined across strata then the adjustment is performed. Overall, both methods provide similar statistical tests, while the win probability approach provides also treatment effect measure with its confidence interval which corresponds to the mentioned statistical test. 
\end{remark}

\section{Applications to a clinical trial data}\label{DAPA}

In this section we will apply the win ratio and rank ANCOVA methodology described in previous sections to the DAPA-HF trial data. Specifically, the data from a PRO (patient reported outcome) questionnaire, which measures heart failure (HF) related symptoms, will be used.

\subsection{Kansas City Cardiomyopathy Questionnaire (KCCQ)}

The KCCQ is a self-administered disease specific instrument for patients with HF (see \cite{Green2000}, \cite{Spertus2005}). The KCCQ consists of 23 questions measuring, from the patients’ perspectives, their HF-related symptoms, physical limitations, social limitations, self-efficacy, and health-related quality of life over the prior 2 weeks. All items are measured on a verbal response scale with 5–7 response options. There are five individual subscales, and all, except the symptom stability question and self-efficacy subscale, are aggregated into a clinical summary score (CSS) (average of the ‘physical limitation score’ and ‘total symptom score’) and overall summary score (OSS) (average of the ‘physical limitation score’, ‘total symptom score’, ‘quality of life score’ and ‘social limitation score’) (Figure \ref{P_KCCQ}). 
\begin{figure}[H]\caption{Mapping of KCCQ items and scores to conceptual domains and summary scores}
\includegraphics[width=10cm]{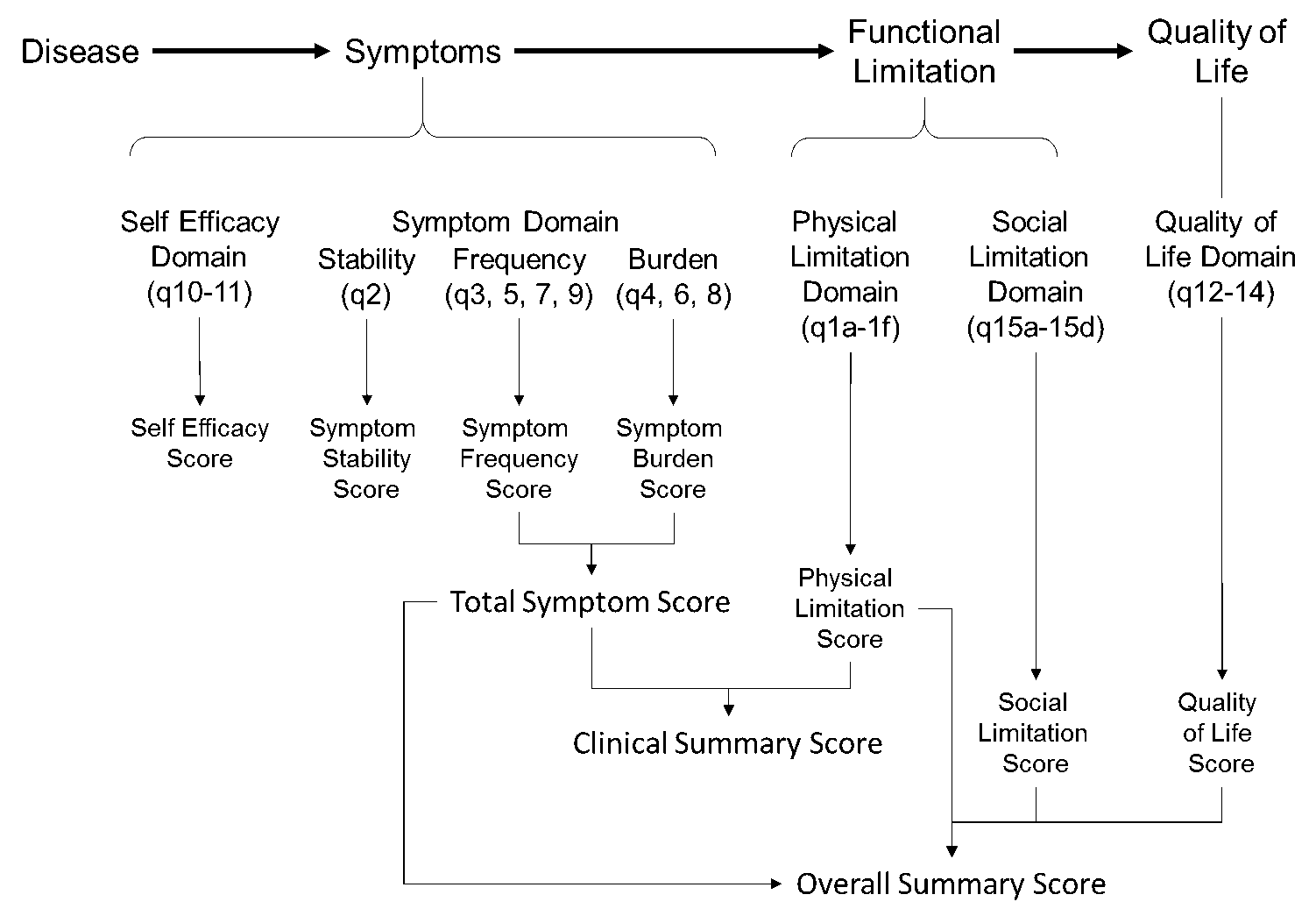}\label{P_KCCQ}
\end{figure}
Scores for each subscale are standardized to range from 0 to 100 with higher scores indicating a better outcome. The KCCQ is scored by assigning each response an ordinal value, beginning with 1 for the response that implies the lowest level of functioning, and summing items within each domain. Scale scores are transformed to a 0 to 100 range by subtracting the lowest possible scale score, dividing by the range of the scale and multiplying by 100. For the analysis we will use only the KCCQ-TSS (Total Symptoms Score) which is the average of the symptoms frequency score and the symptoms burden score (see Figure \ref{P_KCCQ}.)

\subsection{DAPA-HF trial design}
DAPA-HF was an international, multicentre, parallel group, event-driven, randomized, double-blind, clinical trial in patients with chronic HFrEF (heart failure with reduced ejection fraction), evaluating the effect of dapagliflozin 10mg, compared with placebo, given once daily, in addition to standard of care, on the risk of worsening heart failure and cardiovascular death (registration number NCT03036124 in ClinicalTrials.gov). The total number of subjects was $N=4744,$ of which $n_2=2373$ were randomized to the dapagliflozin group, $n_1=2371$ to the placebo group. In the hierarchical testing procedure for endpoints under strong type I error control, the third  secondary endpoint was the change from baseline measured at 8 months in the KCCQ-TSS. The KCCQ-TSS was assessed at baseline (randomization) and at 4 and 8 months after randomization and was analyzed as a composite, ordinal variable, incorporating the vital status of subjects at 8 months along with a change in score from baseline to 8 months in surviving subjects, while missingness for reasons other than death was imputed using the multiple imputation method under the Missing At Random assumption. The treatment effect was estimated using the win ratio appraoch. The analysis yielded a win ratio of $1.18\ \ (1.11,1.26),\ \ p<0.0001$ (see \cite{JMM2019}.) The win ratio was calculated using the adjusted win probability approach with stratification, as described in Section \ref{AWPS}. The statistical test of the null hypothesis, on the other hand, was performed using the rank ANCOVA approach described in Section \ref{RANK}. In the subsequent sections we give more details on methods for calculating the win ratio and the statistic used for the hypothesis testing.

\subsection{Complete case analysis}
In total $N=3891$ subjects had the KCCQ change from baseline at month 8 measured, i.e. both baseline value and value at month 8 was available ($n_1=1965$ in the placebo group and $n_2=1926$ in the dapagliflozin group). Table \ref{Tab1} gives the details of the mean difference analysis of KCCQ-TSS scores between the two groups.

\begin{table}[h!]\centering
\caption{Change from baseline at month 8 in KCCQ-TSS, unadjusted complete case analysis without stratification}
\ra{1.3}
\begin{tabular}{@{}cccccccccccc@{}}\toprule
& \multicolumn{3}{c}{Treatment} & \multicolumn{3}{c}{Comparison}\\
\cmidrule{2-4} \cmidrule{6-7} 
&& Dapa  &  Placebo  &  \\ 
&& $n_2=1926$  &  $n_1=1965$  &   \\ \midrule
mean (sd) && 6.1, (18.6)  & 3.3, (19.2)&&\\
mean diff (CI), t-test && &&  2.8 (1.6,\,4), $p<0.0001$ &\\
WR (CI), non-parametric test && &&  1.21 (1.13,\,1.3), $p<0.0001$ &\\
parametric WR && &&  1.18  &\\
\bottomrule
\end{tabular}
\label{Tab1}
\end{table}
Figure \ref{P5} shows the histograms of the change from baseline at month 8 in KCCQ-TSS score. They indicate that it is reasonable to assume normality of the underlying distributions in both treatment groups for the change from baseline of the KCCQ score. Therefore, a t-test can be performed to compare the treatment effect in two groups.
\begin{figure}[H]\caption{Histograms of change from baseline at month 8 in KCCQ-TSS by treatment group}
\includegraphics[width=12cm]{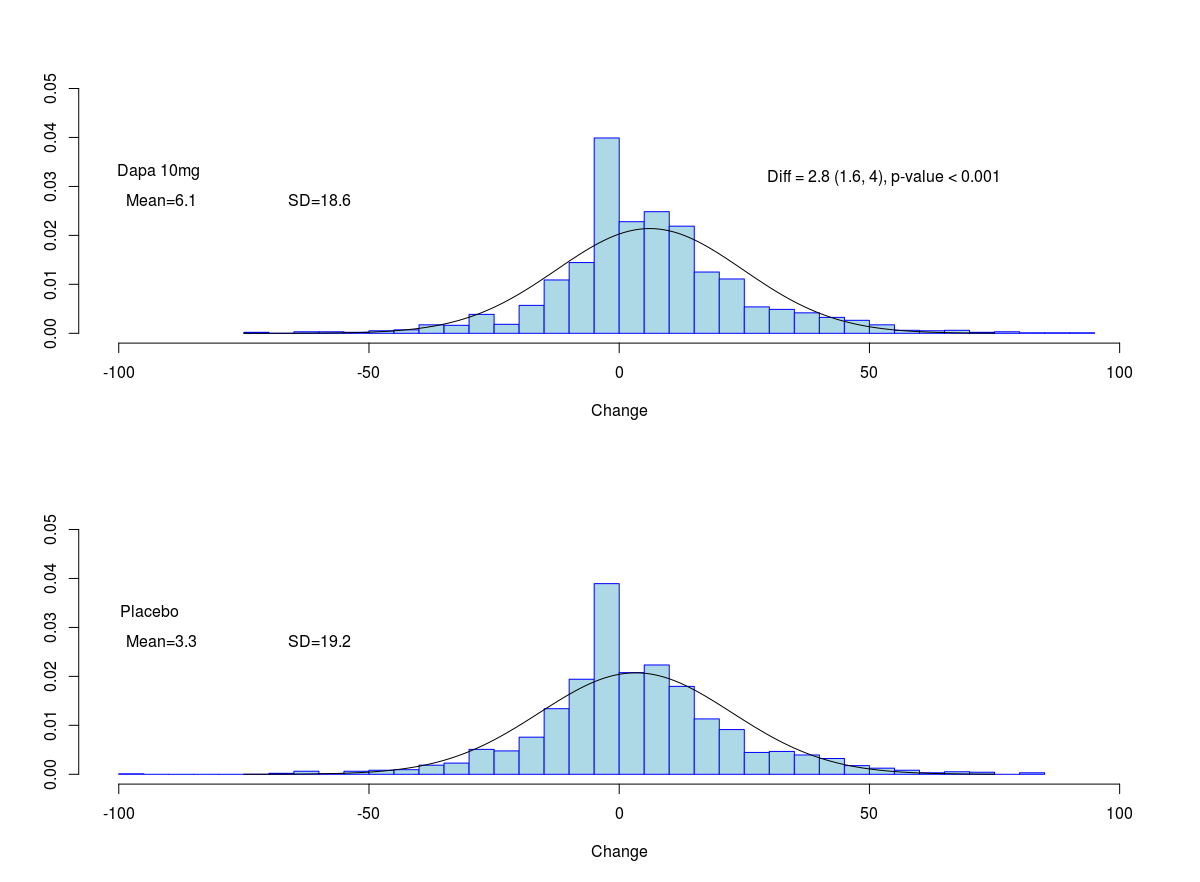}\label{P5}
\end{figure}
The estimates presented in Table \ref{Tab1} are for the win ratio \eqref{Eq12} and the p-value corresponds to the hypothesis test \eqref{Eq13}. Assuming normal distributions and estimating the mean and standard deviation of these distributions from the data (see Table \ref{Tab1}), we get $\xi\sim\N(3.3,19.2^2)$ and $\eta\sim\N(6.1,18.6^2)$. Here again $\eta$ denotes the random variable describing KCCQ-TSS in the active (dapagliflozin) group, whereas $\xi$ denotes the KCCQ-TSS in the placebo group. This means that we can calculate the parametric estimate (based on the form of the distribution) of the win ratio using the formula \eqref{Eq10} for the win probability. A non-parametric estimate for the win ratio can be constructed using formula \eqref{Eq5}, while the confidence interval and the p-value of the null hypothesis that the win ratio equals to 1 can be calculated using Theorem \ref{T1}. The results are summarized in Table \ref{Tab1}.

Below we provide the two-way ANCOVA analysis for the change from baseline at month 8 in KCCQ-TSS adjusted for the baseline KCCQ-TSS score and including the type 2 diabetes status at baseline as a stratification factor (see \cite{Kosi2019}). The least-square means and the corresponding test of their equality are summarized in Table \ref{Tab2}.

\begin{table}[H]\centering
\caption{Change from baseline at month 8 in KCCQ-TSS, adjusted complete case analysis with stratification}
\ra{1.3}
\begin{tabular}{@{}cccccccccccc@{}}\toprule
& \multicolumn{3}{c}{Treatment} & \multicolumn{3}{c}{Comparison}\\
\cmidrule{2-4} \cmidrule{6-7} 
&& Dapa  &  Placebo  &  \\ 
&& $n_2=1926$  &  $n_1=1965$  &   \\ \midrule
lsmean (stderr) && 6.0, (0.37)  & 3.4, (0.37)&&\\
lsmean diff (CI), t-test && &&  2.8 (1.9,\,3.6), $p<0.0001$ &\\
WR (CI), adj-str test && &&  1.2 (1.12,\,1.28), $p<0.0001$ &\\
\bottomrule
\end{tabular}
\label{Tab2}
\end{table}

\subsection{Handling missingness for reasons other than death}
Missing data was categorized into two categories; due to death and not due to death. Any KCCQ-TSS value which was not missing due to death was imputed using multiple imputation (see \cite{Rubin2004}) under the assumption that it was Missing At Random (MAR). The imputation was done sequentially, as described below.

First, all occurrences of missing data where there were observations made after the missing data point (non-monotone missingness), were replaced in multiple imputation datasets, using Markov Chain Monte Carlo with separate chains per subject. In this imputation model, only randomization stratum, treatment arm and observed KCCQ values were included.

For missing data where there were no observations made after the missing data point (monotone missingness), a predictive mean matching imputation approach was applied using the posterior predictive distribution. A linear regression model was estimated, based on observed data and including randomization stratification factor, treatment arm, previously observed values and number of preceding heart failure hospitalizations, as predictors. New predictions were obtained for the observed data points, using the estimated regression coefficients. The same linear model was then used to find the posterior predictive distribution of the regression coefficients. From this distribution, new regression coefficients were randomly drawn and, using the newly drawn regression coefficients, predicted values were obtained for the missing data points. For each missing data point, its five closest neighbors were identified among the predicted values of the originally observed data points. An imputation value was randomly selected from these five values.
This was done sequentially, starting at the earliest time point and progressing until the last time point had been imputed for all subjects with missing data, including imputations made along the way. This procedure is done multiple times, creating multiple imputation datasets. The analysis results on imputed data are then pooled across imputation datasets, taking into account both within-dataset variation and between-dataset variation.

\begin{remark}\label{R15} A simple method of handling missingness for reasons other than death would be to apply a sequential monotone method beginning with imputing missing data for a first post-baseline visit and subsequently proceeding to impute a second post-baseline visit from preceding observed or imputed values. Alternative methods to multiple imputation method not requiring the missing at random assumption are described in \cite{Fan2016}. One other possibility (implemented in \cite{Kaw2015}) is to manage missing values as being tied with all other values, and this way of proceeding is applicable to both deaths and missing values for other reasons; and its invocation enables assessment of treatment comparisons in an environment which is reasonably neutral with respect to deaths as well as other missing values and also does not involve a missing at random assumption.
\end{remark}

\subsection{Incorporating death}
In Section \ref{WRA} we discussed three strategies of incorporating deaths into the analysis of symptom scores. Here we will deal with only the first strategy, that is, subjects having experienced death before the assessment date of the symptom score will be assigned the same worst (lowest) ordinal value. The missing values not due to death were considered as missing at random and were imputed using the multiple imputation method. Hence, the number of non-missing change from baseline values at month 8 was 4744 minus the number of deaths prior to that time point. Overall 257 deaths happened prior to month 8 (121 in the dapagliflozin group and 136 in the placebo group), hence making the number of available changes from baseline (including the imputed values of KCCQ-TSS) $2235$ in the placebo group and $2252$ in the dapagliflozin group. Table \ref{Tab3} summarizes the results for the adjusted win ratio estimation with stratification (see Section \ref{AWPS}) and the test based on the rank ANCOVA (see Section \ref{RANK}).

\begin{table}[H]\centering
\caption{Adjusted analysis of the composite of KCCQ-TSS and death}
\ra{1.3}
\begin{tabular}{@{}cccccccccccc@{}}\toprule
& \multicolumn{3}{c}{Treatment} & \multicolumn{3}{c}{Comparison}\\
\cmidrule{2-4} \cmidrule{6-7} 
&& Dapa  &  Placebo  &  \\ 
&& $n_2=2373$  &  $n_1=2371$  &   \\ \midrule
WP && 0.54  & 0.46&&\\
WR (CI) && &&  1.18 (1.11,\,1.26)&\\
Adj-str WR test && && $p<0.0001$ &\\
rank ANCOVA && && $p<0.0001$ &\\
\bottomrule
\end{tabular}
\label{Tab3}
\end{table}

\subsection{Discussions}
In Section \ref{WRA} we described several strategies to handle intercurrent events in the analysis of change from baseline in the symptoms scores.  In the win ratio analysis in the DAPA-HF study, all intercurrent events, except deaths, were handled using the treatment policy strategy, i.e. these events are disregarded and subjects were followed as if the events had not occurred. As is described in \cite{ICH2019}, the treatment policy strategy cannot be implemented for intercurrent events that are terminal events, since values for the variable after the intercurrent event do not exist. Hence the composite strategy was used to handle deaths. If we were to combine all HFH intercurrent events in an endpoint, that would have more impact on the estimand, and instead of estimating the effect of the treatment on the change in symptoms score we would be estimating a ``net clinical benefit". The same would be true if we were to incorporate death into the composite endpoint using scenario 3 of Section \ref{WRA}, which uses a comparison of deaths based on a characteristic directly related to the event of death, for example the timing of death. Therefore only scenarios 1 and 2 were considered in the DAPA-HF study. By defining an order between deaths based on a characteristic not directly related to the event of death, a separation of the effect of the treatment on the death and on the symptoms scores was done, so the effect on symptoms scores was not driven by the effect of the treatment on risk of death.

The results presented in Tables \ref{Tab1}, \ref{Tab2} and \ref{Tab3} show that the estimated treatment effect on the symptoms score is robust, that is, it is not contingent upon the choice of statistical methods and assumptions. The estimation of the adjusted win ratio with stratification does not have any distributional assumptions, and the corresponding statistical test is similar to the rank ANCOVA test. In the complete case analysis for the change from baseline, the mean difference is an appropriate method to describe the difference in distributions, since Figure \ref{P5} demonstrates the normality (hence symmetry) of underlying distributions. Because of normality of underlying distributions, the estimated non-parametric win ratio and the parametric win ratio in Table \ref{Tab1} are almost the same. In a setting where the underlying distributions are not normal, the non-parametric win ratio estimate will still be valid. Moreover, combining the numerical changes from baseline with death as the worst possible change transforms the variables of interest into ordinal variables, and the win ratio is, again, an appropriate method to test the difference in distributions. The win ratio in Table \ref{Tab2} which is a complete case analysis and the win ratio in Table \ref{Tab3} which is an imputation based analysis (with multiple imputation of missing data not due to death and including the events of death), show that although adding deaths into the analysis yields a more complete estimate of the treatment effect, in a more realistic setting where subjects can die, the magnitude of treatment effect is the same, which confirms that the effect in KCCQ-TSS scores is not driven by the treatment effect on reducing mortality, which was observed in the trial, as shown in \cite{JMM2019}. The method of handling deaths in Table \ref{Tab3} corresponds to scenario 1, that is, all deaths were treated equal. This analysis was the sensitivity analysis in the DAPA-HF trial, whereas the primary method of analysis was based on the scenario 2, where an order between deaths was introduced based on the last change from baseline of the subject while alive. That method yields the same estimate and the confidence interval, confirming that, in this case, the handling of deaths does not change the treatment effect estimate of KCCQ-TSS scores. This was due to the fact that although there was an effect in all-cause mortality for the full duration of the study, at month 8 the number of deaths was small and was balanced in both treatment groups. In case where the difference in number of deaths in both treatment groups is clearly different at the timepoint of measurement, the handling of death can have more influence on the estimated treatment effect.

One potential drawback of the win ratio (and the associated win probability) is that the statistical interpretation of this measure is not as straightforward as, say, a difference in means. Both measures are essentially group-level estimates; making inference regarding the similarity of two independent groups, with respect to the location of their distributions. When looking at the difference in means, this is interpreted as the average difference in locations, assuming that the underlying variable is continuous and normal, and that there are no intercurrent events which influence the estimand (what is being estimated) if they are ignored. Average change in each group of symptoms scores is easy to interpret since it has the same unit of measurement as the symptoms score itself, meaning that participants of each group would expect to have approximately the same change in symptoms as the average of their group. If the underlying data is not continuous, or there are major departures from normality, a non-parametric approach is more appropriate. The non-parametric win ratio can also easily incorporate a composite strategy for handling intercurrent events (e.g. deaths). The interpretation of the win ratio is that it is the average odds of the win probability, i.e. the chance of one group having a ``greater benefit" compared to the other. The estimated number of subjects who need to be treated to observe such a benefit, can be calculated and expressed using the win probability. This would amount to deriving a Number Needed to Treat (NNT) based on the win ratio.

To better understand the treatment effect that corresponds to the win ratio of 1.18 (or, equivalently, the win proportion 0.54, see Table \ref{Tab3}), we can calculate the estimated NNT as defined in Remark \ref{R16}. The win proportion equal to 0.54 will yield an $\text{NNT}=\frac{1}{2*0.54-1}=13$ (to get an integer in calculating the NNT we always round up fractional numbers), which can be interpreted as 13 subjects need to be treated by dapagliflozin in addition to standard of care, compared to being treated with standard of care alone, to have one subject with better benefit in symptoms. It is important to note that the win ratio is calculated based on the change from baseline. The least-squares mean of the change from baseline in the placebo group is 3.4 (see Table \ref{Tab2}), which shows that in the placebo group, as well as in the dapaliflozin group, there is an improvement in symptoms. Hence, to be precise, 13 subjects need to be treated by dapagliflozin to have one subject with better improvement in symptoms than they would have if they were treated only by standard of care.

\begin{table}[H]\centering
\caption{NNT based on the win probability}
\ra{1.3}
\begin{tabular}{@{}cccccccccccc@{}}\toprule
& \multicolumn{3}{c}{Comparison} & \multicolumn{3}{c}{NNT}\\
\cmidrule{2-4} \cmidrule{6-7} 
&& Win ratio  &  Win prob  &&  \\  \midrule
&& 1.05&	0.5121951&& 	41\\
&& 1.1&	0.5238095&& 	21\\
&& 1.15&	0.5348837&& 	15\\
&& 1.18&	0.5412844&& 	13\\
&& 1.2&	0.5454545&& 	11\\
&& 1.25&	0.5555556&& 	9\\
&& 1.3&	0.5652174&& 	8\\
&& 1.35&	0.5744681&& 	7\\
&& 1.4&	0.5833333&& 	6\\
&& 1.45&	0.5918367&& 	6\\
&& 1.5&	0.6&& 	5\\
&& 2&	0.6666667&& 	4\\
&& 3&	0.75&&	2\\
&& -&	1&&	1\\
\bottomrule
\end{tabular}
\label{Tab4}
\end{table}

\section{Conclusions}\label{conclusions} The win ratio is a general method of comparing locations of distributions of two independent, ordinal random variables. Under minimal assumptions, an asymptotically normal estimator for the win ratio can be derived. Stratification and numeric covariate adjustment can be made using simple modifications of the estimator for the crude win ratio. It was shown that the win ratio and its modifications for stratified analysis, adjusted analysis and adjusted analysis with stratification give tests that, under the null hypothesis, are correspondingly equivalent to the Wilcoxon rank-sum test or to the Fligner-Policello test, linear regression testing on the ranks, van Elteren or Cochran-Mantel-Haenszel test and the rank ANCOVA test, with the advantage that the win ratio itself provides an interpretable treatment effect measure. Hence, the unified method of the win ratio described in the present work can be used instead of several different tests both for treatment effect estimation and for the corresponding treatment effect difference testing.

\addcontentsline{toc}{section}{References} 

\end{document}